\documentclass[journal,10pt]{IEEEtran}
\usepackage[noadjust]{cite}
\usepackage{color}
\usepackage{amsthm}
\usepackage{amsmath}
\usepackage{amssymb}
\usepackage{array}
\usepackage{algorithm}
\usepackage{caption}
\usepackage{graphicx}
\usepackage{subcaption}
\usepackage{algpseudocode}
\usepackage{tikz,pgfplots,pgfplotstable}
 \usetikzlibrary{patterns}
 \usetikzlibrary{calc}
 \usetikzlibrary{shapes,arrows,positioning}
\newtheorem{defn}{Definition}
\newtheorem{lemma}{Lemma}
\newtheorem{assum}{Assumption}
\newtheorem{theorem}{Theorem}
\newtheorem{remark}{Remark}
\begin{document}
\title{Completely Uncoupled Algorithms for \\ Network Utility Maximization }
\author{\IEEEauthorblockN{S.Ramakrishnan\IEEEauthorrefmark{1},Venkatesh Ramaiyan\IEEEauthorrefmark{2} \\}
\IEEEauthorblockA{Department of Electrical Engineering, IIT Madras\\ Email: \IEEEauthorrefmark{1}ee12d036@ee.iitm.ac.in, \IEEEauthorrefmark{2}rvenkat@ee.iitm.ac.in,
}\thanks{A part of this work appeared in the proceedings of International Conference on Signal Processing and Communications (SPCOM), Bangalore, 2016, pp. 1-5. DOI:10.1109/SPCOM.2016.7746656}}

\maketitle
\algnewcommand{\Initialize}[1]{ 
  \State \textbf{Initialize:}
  \Statex \hspace*{\algorithmicindent}\parbox[t]{.8\linewidth}{\raggedright #1}
}
\begin{abstract}
In this paper, we present two completely uncoupled algorithms for utility 
maximization. In the first part, we present an algorithm that can be 
applied for general non-concave utilities. We show that this algorithm induces a 
perturbed (by $\epsilon$) Markov chain, whose stochastically stable 
states are the set of actions that maximize the sum utility. 
In the second part, we present an approximate sub-gradient algorithm 
for concave utilities which is considerably faster and requires lesser memory. 
We study the performance of the sub-gradient algorithm
 for decreasing and fixed step sizes. We show that, for 
decreasing step sizes, the Cesaro averages of the utilities converges to a 
neighborhood of the optimal sum utility. For constant step size, we show that the 
time average utility converges to a neighborhood of the optimal sum utility. 
Our main contribution is the expansion of the achievable rate region,
which has been not considered in the prior literature on completely uncoupled 
algorithms for utility maximization.
This expansion aids in allocating a fair share of resources  to the nodes which
is important in applications like channel selection,
user association and power control.  
\end{abstract}
\section{Introduction}
Radio resource allocation is an important problem in infrastructure, ad-hoc and 
sensor networks \cite{georgiadis2006resource}. In particular we need to address 
the  following resource allocation problems, viz channel selection, user association 
and power control. Channel selection and power control are essential for the efficient
 use of radio resources, whereas user association deals with efficient use of
 deployed Access Points. The solution should cater to the following
 objectives: $(i)$ Network throughput optimality be ensured $(ii)$ Users get a fair
 share of the network throughput $(iii)$ Ease of implementable. 
 With the advances in 5G wireless systems, it is predicted
  that there will be a phenomenal
 increase in the number of access points \cite{6824752}. Added to this, we have 
 coexisting radio technologies like LTE and Wifi \cite{7143339}. 
In such scenarios, centralized solution is unsuitable due to a large overhead.
Further, centralized control is also impractical in a heterogeneous setup.
Thus robust and easy to implementable distributed solutions are desirable.

\par \indent In this paper, we provide solutions to the above problems with the 
stated objectives. We arrive at such a solution using the approach in
 \cite{pradelski2012learning}, which the authors call as completely uncoupled 
 learning. In completely uncoupled learning, nodes' decisions are based only on their 
 past actions and utilities. In \cite{marden2014achieving}, Marden et al propose a 
 completely uncoupled algorithm that maximizes the sum-payoff (which translates to maximizing sum-throughput in wireless networks). 
 An important attribute to consider in radio resource allocation is fairness 
 among nodes, i.e. every node should get a fair share of the network  throughput.  
In this paper, we consider the problem of  utility maximization, 
where maximizing some utility functions have a notions of fairness \cite{5461911}.
 We propose two completely uncoupled algorithms that 
 maximize the sum-utility of the nodes.  
In our first algorithm, we discretize the rate region, thereby we pose
the utility maximization as a combinatorial optimization problem.
 This algorithm applies to general utilities, not necessarily concave.  
 In our second algorithm, we propose an approximate sub-gradient algorithm 
 for maximizing concave utilities. The main contribution of our work is to 
 provide flexibility in operating at any point in the interior of the rate region.

 Our algorithms are general and can be applied to any general network utility 
 maximization not restricted to wireless networks. We present our 
 algorithm for a general network, while bearing the above stated applications in mind.

\subsection{Related Literature}
Tassiulas and Ephremides proposed the Max-weight algorithm in 
\cite{182479}. The Max-weight algorithm 
 can stabilize any arrival rate within the rate region \cite{182479}. 
 Proportional fair scheduler was shown to optimize logarithmic utility
  function in \cite{1310314}. 
 In \cite{4455486}, Neely et al proposed an algorithm that could stabilize
 any arrival rate within the rate region and optimize a concave utility for arrival rates
 exceeding the rate region. The main drawback of the Max-weight algorithm, used in 
\cite{182479,1310314,4455486}, is its complexity and 
centralized nature. 
 
\par Maximal scheduling algorithms, having low complexity,  
  could support only a fraction of the rate region \cite{4439837}. Greedy algorithms 
  such as longest queue first scheduling, are optimal only for a class of network topologies \cite{Joo:2009:UCR:1618562.1618572}, 
  \cite{Leconte:2011:IBT:2042972.2042980}. 
 Distributed algorithms based on Gibbs sampling were proposed for IEEE 802.11 WLANs  in \cite{4215753}, for channel selection and user association. 
A proportional fair resource allocation algorithm for channel selection and user 
association was proposed in \cite{6636117}. Both \cite{4215753} and \cite{6636117} 
 require neighbor information exchange (or knowledge) and are applicable only to 
some tailored utilities.
 
 \par In \cite{DBLP:journals/ton/JiangW10}, Jiang and Walrand  proposed 
distributed scheduling algorithms for a conflict graph model without collisions.  
 They proved that their algorithms are optimal assuming time scale separation. 
 In \cite{liu2010towards}, Liu et al showed that stochastic approximation 
  \cite{borkar2006stochastic} leads to time scale separation when 
   the update parameters are bounded. In \cite{5625654}, Jiang et al  proposed 
  distributed scheduling algorithms and showed them to be optimal without 
  time scale  separation or bounded parameters assumption.  
  A discrete time version called Q-CSMA was proposed by Jian Ni et al in 
  \cite{6097082},  where collision free schedules are generated by considering a 
  control phase, thereby allowing multiple links to change their state.   
  SINR model was assumed in \cite{7145489}, where the authors show that 
  any arrival rate in the interior of the rate region could be supported. 
 


\par Resource allocation problems have been studied as cooperative and 
 non-cooperative games. In repeated prisoners dilemma with selfish players, 
attempts were made to induce cooperation in 
\cite{Kraines1989,axelrod1981evolution,nowak1993strategy}.  
In \cite{nowak1993strategy}, Pavlov method was proposed, which is 
win-stay lose-shift. In our model, we do not assume that the nodes (players) 
are selfish. Rather, we assume that the nodes (players) cooperate,
however are restricted in information, i.e. they do not know the actions or 
payoffs of other nodes (players). A similar assumption is used 
in \cite{pradelski2012learning,marden2014achieving} for maximizing the sum
payoff in a completely distributed  manner.
 
\par In \cite{MONDERER1996124}, Monderer and Shapley showed the existence
 of Nash equilibrium for potential games and  convergence of better reply dynamics to
 Nash equilibrium. Blume studied the interactions of players residing on an infinite 
 lattice  in \cite{BLUME1993387}, where players choose actions based on best 
 response, perturbed best response and log-linear learning rules. In \cite{4814554}, 
 Marden et al  formulated cooperative control problems as a repeated potential game. 
 They designed objective functions for players, such that the log-linear learning 
 rule converges to a pure Nash equilibrium in a probabilistic sense, where Nash 
 equilibrium action is played for a large fraction of time. 
 The idea of state based potential games was introduced in \cite{MARDEN20123075} 
 with the goal of designing local objectives to attain a desired global objective 
 through log-linear learning, where the introduction of state  helps in coordination. 
 In \cite{6459524}, Li and Marden considered designing local objectives for 
  state based potential games with continuous action sets with the objective of 
  minimizing a convex function of the joint action profile. 

 \par A completely uncoupled algorithm to reach efficient Nash equilibrium 
  was proposed by Pradelski et al in \cite{pradelski2012learning}. 
  The algorithm was based on the  theory of perturbed Markov chains 
  \cite{freidlin2012random},\cite{RePEc:ecm:emetrp:v:61:y:1993:i:1:p:57-84}.
  With similar ideas, Marden et al proposed algorithms to achieve maximum sum 
  payoff in  \cite{marden2014achieving}. These algorithms were adapted to wireless 
  networks in \cite{6576468},\cite{6573240}. In \cite{7040463}, Borowski et al 
  proposed distributed algorithm to achieve efficient correlated equilibrium. 
  In our prior work \cite{7746656}, we proposed a distributed algorithm for 
  utility maximization and used perturbed Markov chain ideas to prove optimality. 
  In \cite{7925713}, we extended \cite{7746656} to state based models.

\subsection{Contributions}
\begin{enumerate}
\item In this paper, we propose two distributed algorithms for utility 
maximization. To the best of our knowledge, 
we are the first to propose completely uncoupled utility maximization 
algorithms that achieves the entire rate region. These algorithms find application in 
distributed channel selection, user association and power control in a variety of 
wireless networks.  
 
\item In the first algorithm, which we call General Network Utility Maximization
 (G-NUM), we allow the utilities to be general functions
 (not necessarily concave) of the average payoff. We show that G-NUM is optimal
 and the sum payoff maximizing algorithm in \cite{pradelski2012learning} is 
  a special case of it. 
   
 \item For concave utilities, we propose our second algorithm,
  Concave Network Utility Maximization (C-NUM). In C-NUM, we present 
   an approximate subgradient algorithm inspired by Gibbs sampling based 
  Utility maximization algorithm in \cite{DBLP:journals/ton/JiangW10}.
  With C-NUM, we show an important connection between completely uncoupled  
  algorithms based on perturbed Markov chains and Gibbs sampling based utility 
  maximization algorithms such as~\cite{DBLP:journals/ton/JiangW10}.   

 \item We also derive upper bounds on the mixing time for the algorithm in 
 \cite{pradelski2012learning}  and show that the mixing time (upper bound) 
 grows exponentially in the number of nodes. 
\end{enumerate}

\subsection{Outline}
The rest of the paper is organized as follows. In Section \ref{section System Model}, 
we discuss the system model. In Section \ref{section: Alg1}, we propose our first algorithm on general 
utility maximization and discuss convergence results. 
We present the second algorithm for concave utilities in Section \ref{Section: Alg2}
with convergence results. We discuss numerical results in Section \ref{section 
Numerical examples}.  In Section \ref{section:Summary and Comparisons}, we 
provide a summary with comparisons.
The proofs of our results are discussed in detail in the Appendix
\ref{section Appendix}.

\section{System Model}
\label{section System Model}
We consider a system of $N$ nodes. We assume a slotted time model. In every time slot 
$t$, node $i$ chooses to play an action $a_i(t) \in \mathcal{A}_i$. We assume that, $\forall i$, $\mathcal{A}_i$ is finite. Let $a(t)=(a_1(t),a_2(t),\cdots,a_N(t)) \in \mathcal{A}$ denote the action profile at time $t$, where $\mathcal{A}=\prod_i \mathcal{A}_i$. In slot $t$,  node $i$ gets a reward $r_i(t)$. We assume that, 
\begin{align*}
r_i(t) = f_i(a(t)),
\end{align*}
where $f_i(\cdot)$ is a non-negative function from $\mathcal{A} \to \mathbb{R}^+$. We allow $f_i(\cdot)$ to be general, thereby allowing our setup  to be applied in a variety of models. 

Let $p(a)$ denote the fraction of time action profile $a \in \mathcal{A}$ is chosen, where $\sum_a p(a) = 1$. The average payoff received by node $i$ is given by,
\begin{align}
\label{average payoff}
\bar{r}_i(p) = \sum_a p(a) r_i(a).
\end{align}
Let $\bar{r}=\left(\bar{r}_1,\bar{r}_2,\cdots,\bar{r}_N\right)$ denote the vector of average payoffs obtained by the nodes. We say that a payoff vector $\bar{r}$ is achievable if there exists a $p>0$ satisfying \eqref{average payoff}. We denote by $\mathcal{R}$, the set of achievable payoffs.  Then,
\begin{align*}
\mathcal{R}= \left\{ \bar{r} = (\bar{r}_i(p)) \mid   \sum_a p(a) =1, \; p(a)>0 \right\}.
\end{align*} 
Let $U_i$ denote the utility of node $i$, where $U_i$ is a function of $\bar{r}_i$. Without loss of generality, we assume that $U_i(\bar{r}_i)$ is bounded between $0$ and $1$. The objective here is to maximize the sum utility $\sum_i U_i(\bar{r}_i)$ where $\bar{r} \in \mathcal{R}$. We formulate the utility maximization problem as,
\begin{align}
\label{utility maximization formulation}
\begin{aligned}
\max &\sum_i U_i(\bar{r}_i) \\
\text{s.t. } &\bar{r}_i \leq \sum_a p(a) r_i(a) \\ 
&\sum_a p(a) =1, \; p(a)\geq0.
\end{aligned}
\end{align} 
In this work, we seek a distributed algorithm that solves~\eqref{utility maximization formulation}. 
\textit{By distributed, we mean that our 
algorithm is completely uncoupled, where every node has knowledge only about its 
previous actions and payoffs.} A node chooses its action purely based 
on its previous actions and payoffs.   
 We assume that the network satisfies the following interdependence 
definition. 
\begin{assum}{Interdependence:}
\label{assumption Interdependence}
For any subset of the nodes $\mathcal{S}$ and any action profile $a = (a_S,a_{-S})$, there exists a node $j \notin S$ and an action profile $({a}^{\prime}_{\mathcal{S}},a_{\mathcal{-S}})$,  such that $f_j(a_S,a_{-S}) \neq f_j(a^{\prime}_S,a_{-S})$.
\end{assum}
\begin{remark}
\label{remark_csma}
\par The adaptive CSMA models in \cite{DBLP:journals/ton/JiangW10,6097082} 
assume a conflict graph based network model. 
Extensions of \cite{DBLP:journals/ton/JiangW10,6097082} to more practical
SINR based interference model was considered in \cite{6256765,6576230}. 
In the above works, the service rate of a link depends 
on the actions of other links only through the notion of feasible actions
(transmission modes) i.e. $r_i(a) = f_i(a_i) \; \text{s.t. } a_i \text{ is feasible}$. 
In contrast, we allow payoffs to be a function of the joint action profile, i.e.
$r_i(a)=f_i(a)$, where $a=(a_1,\cdots,a_N)$ is the joint action profile.
Such an assumption is preferred in applications such as user association and 
channel selection (See models in \cite{hou2013proportionally}
\cite{borst2014nonconcave}).  In our work, we assume that the 
network satisfies interdependence which enables us to work with a more general 
model.
\end{remark}

%

\section{General Network Utility Maximization algorithm}
\label{section: Alg1}
In this section, we present a completely uncoupled utility maximization algorithm 
for general utilities (possibly non concave) and discuss its convergence results. 
Algorithm \ref{Alg1}, which we call General Network Utility Maximization 
(G-NUM) algorithm is described below. 

\begin{algorithm}
\caption{\textbf{: General Network Utility Maximization Algorithm (G-NUM) }}
\label{Alg1}
\begin{algorithmic}
\Initialize{Fix $c > N$, $\epsilon > 0$ and $K \geq 1$. \\ For all $i$, set $q_i(0) =0$.}
\end{algorithmic}

\begin{algorithmic}
\State \textbf{Action update at time $t$:}
\If {($q_i(t-1)=1$)}
\State
$a_i(t) =
\begin{cases}
a_i(t-K) & \text{w.p. \ } 1 - \epsilon^c \\
a_i \in {\mathcal A}_i & \text{w.p. \ } \frac{ \epsilon^c}{\left|\mathcal{A}_i\right|-1} \text{\ if\ } a_i \neq a_i(t-K)
\end{cases}$
\Else
\State $a_{i}(t)= a_i$ w.p. $\frac{ 1}{\left|\mathcal{A}_i \right|}$ where $a_i \in {\mathcal A}_i$
\EndIf
\end{algorithmic}

\begin{algorithmic}
\State \textbf{Update for $q_i(\cdot)$ at time $t$:}
\If { ($q_i(t-1)=1$) and ($a_i(t)=a_i(t-K)$) \\ and $\left( \sum\limits_{j=t-K+1}^{t}  r_i(j) = \sum\limits_{j=t-K}^{t-1}  r_i(j)\right)$}
\State $q_i(t) = 1 \text{\ w.p. \ } 1$
\Else
\State
$q_i(t) =
\begin{cases}
1 & \text{w.p. \ } \epsilon^{1-U_ i\left(\frac{1}{K}\sum\limits_{j=t-K+1}^{t}  r_i(j)\right)} \\
0 & \text{w.p. \ } 1- \epsilon^{1-U_ i\left(\frac{1}{K}\sum\limits_{j=t-K+1}^{t}  r_i(j)\right)}
\end{cases}$
\EndIf
\end{algorithmic}
\end{algorithm}

The history (and possible state)
of any node $i$ at the end of slot $t-1$ is the sequence of actions 
$(a_i(1),\cdots,a_i(t-1))$
and the payoffs received $(r_i(1),\cdots,r_i(t-1))$.
We require that the nodes maintain an internal ``satisfaction'' variable, 
denoted by $q_i(t-1)$ (at time $t-1$), which is a function of the action and the payoff 
received in the previous $K$ slots (where $K$ is a fixed positive integer).
We let $q_i(\cdot)$ take values from the binary set $\{ 0,1 \}$, where $q_i(\cdot) = 1$
represents a state of ``content'' with the choice of actions and the payoff received (in 
the previous $K$ slots), while $q_i(\cdot) = 0$ represents a state of ``discontent'' for 
the node.  For every slot, the nodes have to choose an 
action $a_i(\cdot)$ and update their satisfaction variable $q_i(\cdot)$.

\par Node $i$ chooses action, $a_i(t)$, at the beginning of slot $t$ depending on its
  satisfaction variable $q_i(t-1)$. If node $i$ is content at the beginning of 
 slot $t$, i.e., $q_i(t-1) = 1$, then it repeats an earlier action, here $a_i(t-K)$,
with high probability $1 - \epsilon^c$ (where, $c$ is a parameter and $c > N$).
 With a small probability $\epsilon^c$, any other action is 
chosen uniformly at random. 
When node $i$ is discontent, i.e., $q_i(t-1) = 0$, the node selects an action 
randomly and uniformly from ${\mathcal A}_i$.

\par The satisfaction variable  $q_i(t)$ is updated by the end of slot $t$.
If the node $i$ was content in slot $t-1$ (i.e., $q_i(t-1) = 1$), then, it continues
 to remain content if it had repeated an earlier action (i.e., if $a_i(t) = a_i(t-K)$, which 
happens with high probability) and received the same payoff in the last $K$ 
slots, i.e., $\sum\limits_{j=t-K+1}^{t}  r_i(j) = \sum\limits_{j=t-K}^{t-1}  r_i(j)$ 
 (which would happen when the action profile in the network remains unchanged).
Otherwise, a node becomes content with a very low probability depending on 
the utility, where the utility is a function of the average of the payoffs received in the 
last $K$ slots.
\begin{figure}[!ht]
    \centering
    \begin{subfigure}[t]{0.475\textwidth}
   \centering
    \begin{tikzpicture}[scale=0.95,every node/.style={scale=0.95}]
\node[inner sep=0](content) at (0,0){\includegraphics[width=.04\textwidth]{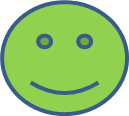}};
 \node[inner sep=0] (C) [below = 0.05cm of content,xshift=0.2cm]{\tiny$q_i(t-1)$};
 \node[inner sep=0](action1)[right = 0.2cm of  content]{{\scriptsize$ a_i(t) =
\begin{cases}
a_i(t-K) & \text{w.p. \ } 1 - \epsilon^c \\
a_i \neq a_i(t-K) & \text{w.p. \ } \frac{ \epsilon^c}{\left|\mathcal{A}_i\right|-1}  
\end{cases}$}};
\node[inner sep=0](discontent) [right = 0.3cm of  action1]{\includegraphics[width=.04\textwidth]{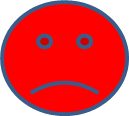}};
 \node[inner sep=0] (D) [below = 0.05cm of discontent,xshift=0.2cm]{\tiny$q_i(t-1)$};
\node[inner sep=0](action2) [right = 0.2cm of  discontent]{{\scriptsize$ a_i(t)=
a_i \text{w.p. \ } \frac{1}{|A_i|}$}};
\draw ($(content.north west)+ (-0.2,0.4)$) rectangle ($(action2.south east)+(0.05,-0.4)$);
\draw ($(action1.north east)+ (0.15,0.18)$) -- ($ (action1.south east)+(0.15,-0.18)$); 
\end{tikzpicture}
\caption{Action Update at time $t$: A content node chooses the action 
that was chosen $K$ slots before with a large probability as indicated  in the left box. 
Whereas, a discontent node chooses an action uniformly 
at random as shown in the right box.}
 \label{Gnum_actionupdate}
    \end{subfigure}  

    \begin{subfigure}[t]{0.475\textwidth}
    \centering
    \begin{tikzpicture}[scale=0.95,every node/.style={scale=0.95}]
      \node[inner sep=0](content) at (0,0){\includegraphics[width=.04\textwidth]{content.png}};
       \node[inner sep=0](content1) at (3,0){\includegraphics[width=.04\textwidth]{content.png}};
      \draw[->,thick] ($(content.east)+(0.1,0)$) to node[above] {{\scriptsize $a_i(t)=a_i(t-K)? $}} node[below]{ {\scriptsize$r_i(t)=r_i(t-K)?$}} ($(content1.west)+(-0.1,0)$); 
       \node[inner sep=0](else)[right= 0.1 cm of content1] {{\scriptsize Else}};
         \node[inner sep=0](content2) [right = 0.4 cm of else,yshift=0.5 cm]{\includegraphics[width=.04\textwidth]{content.png}};
         \node[inner sep=0](discontent) [right = 0.4 cm of else,yshift=-0.5 cm]{\includegraphics[width=.04\textwidth]{discontent.png}};
         \node[inner sep =0](C)[right=0.1cm of content2,yshift=0.2 cm]{{\scriptsize w.p. $\epsilon^{1-U_i\left(\frac{1}{K}\sum\limits_{j=t-K+1}^{t}  r_i(j)\right)}$}};
           \node[inner sep =0](D)[right=0.1cm of discontent,yshift=0.2 cm]{{\scriptsize w.p. $1-\epsilon^{1-U_i\left(\frac{1}{K}\sum\limits_{j=t-K+1}^{t}  r_i(j)\right)}$}};
           \draw[->,thick] ($(else.east)+(0.1,0)$) to ($(content2.west)+(-0.1,-0.1)$); 
           \draw[->,thick] ($(else.east)+(0.1,0)$) to ($(discontent.west)+(-0.1,0.1)$); 
           \draw ($(content.north west)+ (-0.1,1)$) rectangle ($(D.south east)+(0,-0.4)$);
           \draw ($(content1.north east)+ (0.05,1)$) -- ($ (content1.south east)+(0.05,-0.9)$); 
    \end{tikzpicture}
    \caption{Satisfaction Update at time $t$: A content node becomes content with 
    probability $1$, if it repeats the action that was chosen $K$ slots before and receives 
    the same payoff it received $K$ slots before, as shown in the left box. In any other  
    case, a node becomes content with probability shown in the right box. Content and 
    discontent states are shown as  a green happy and red sad smiley respectively.}
 \label{Gnum_satisfactionupdate}
    \end{subfigure}
    \caption{Gnum update: The Figure \ref{Gnum_actionupdate} shows 
    the update of actions and Figure \ref{Gnum_satisfactionupdate} shows the
    update of satisfaction variable at time $t$.}
    \label{Fig:Gnum}
\end{figure}

\par The satisfaction variable aids in synchronizing changes in actions across the network.
When all the nodes are content, i.e., $q_i(t) = 1$ for all $i$, all the 
nodes continue to repeat the last $K$ actions (in synchrony) and continue to receive 
a constant average payoff based on the sequence of actions. Now if a node decides to 
change its action, it becomes discontent with a large probability and chooses its 
action randomly in the subsequent slots. By interdependence this causes
other nodes in the network to become discontent and sets a ripple effect 
causing all nodes in the network to become discontent. 
Finally, the nodes become content again, where a sequence of action profiles
is preferred depending on the average payoff corresponding to the $K$-sequence and 
nodes' utilities. 
In the remainder of this section, we will prove that the above algorithm chooses an 
action sequence
that optimizes the formulation in (\ref{utility maximization formulation}) as 
$\epsilon \rightarrow 0$ and as $K \rightarrow \infty$.
\begin{remark}
In G-NUM, we restrict the rate region to those points achievable by a sequence of $K$
actions. Thereby, the problem of maximizing the utility  
is posed as a combinatorial optimization problem. 
A simple solution to this modified problem is to use a frame with 
a sequence $K$ of actions and apply the algorithm 
in \cite{marden2014achieving} over these frames. Such an approach has been 
used in \cite{7040463} to achieve efficient correlated equilibrium. However, this 
requires an additional requirement of frame synchronization. In G-NUM,  
 we do not use frames consisting of a sequence of actions, instead we let the action
 at time $t$ to depend on the previous $K$ actions, thereby avoiding the requirement 
 of frame synchronization.
\par A sequence of transmit power levels is used to stabilize a set of 
arrival rates in \cite{7145489}, wherein a node becomes content if the node achieves 
its arrival rate. In contrast, in G-NUM, we would like the nodes to maximize the 
utilities. This leads to a significant difference to the analysis of G-NUM as 
compared to the algorithm in \cite{7145489}. 
\end{remark}
\begin{remark}
With $K=1$, G-NUM replicates the algorithm in 
\cite{marden2014achieving}. In this sense, G-NUM
generalizes the ideas presented in works such as  \cite{pradelski2012learning} 
and \cite{marden2014achieving}. 
\end{remark}

\subsection{Performance Analysis of G-NUM}
In this section, we discuss the optimality of G-NUM. We characterize the performance 
of G-NUM as $t \to \infty$ and $\epsilon \to 0$. To analyze the performance of 
G-NUM, we use tools  from perturbed Markov chains \cite{freidlin2012random},
\cite{RePEc:ecm:emetrp:v:61:y:1993:i:1:p:57-84}. 
We first show that G-NUM induces a perturbed Markov chain 
(perturbed by $\epsilon$). In Theorem  \ref{thm optimality of algorithm 1}, 
we show that the stochastically stable states  
(See Definition \ref{defn stochastically stable states}) of the Markov chain 
 induced by G-NUM are the set of actions that maximize the sum utility of the nodes.  
Define $X_{\epsilon}(t)$ as,
\[ X_{\epsilon}(t)=\prod_{i=1}^N (a_i(t-K+1),..,a_i(t),q_i(t)). \]
$X_{\epsilon}(t)$ corresponds to the actions of all the nodes in the previous $K$ slots
and the ``satisfaction'' variable of the nodes in the current slot $t$.
In the following Lemma, we show that $X_{\epsilon}(t)$ is a regular perturbed 
Markov chain (perturbed by the algorithm parameter $\epsilon$) with a positive 
stationary distribution.

\begin{defn}{Regular perturbed Markov Chain:}
\label{defn Regular perturbed Markov Chain}
A Markov process $X_{\epsilon}(t)$, with state space $\Omega$ and transition 
probability $P^{\epsilon}$, is a regular perturbed Markov process (perturbed by $
\epsilon$) if the following conditions are satisfied 
(see \cite{RePEc:ecm:emetrp:v:61:y:1993:i:1:p:57-84}).
\begin{enumerate}
\item $\forall \epsilon > 0$, $X_{\epsilon}(t)$ is an ergodic Markov Process
\item $\begin{aligned}[t]
\forall x, y \in \Omega, \lim_{\epsilon \to 0} P^{\epsilon} ( x,y ) = P^{0}(x,y) 
\end{aligned}$
\item $ \begin{aligned}[t]
\forall x,y \in \Omega,\ \text{if\ } &P^{\epsilon} ( x , y) > 0 \text{ for some } \epsilon > 0, \text{then},
\end{aligned}$
\[ 0< \lim_{\epsilon \to 0} \frac{ P^{\epsilon} (x , y)}{\epsilon^{r(x , y)}} < \infty, \]
and $r(x , y)\geq 0$ is called the resistance of the transition $x \to y$.
\end{enumerate}
\end{defn}

 \begin{lemma}
 \label{regular perturbed Markov chain}
$X_{\epsilon}(t)$ is a regular perturbed Markov chain (perturbed by $\epsilon$)
over the state space $\Omega =(\mathcal{A}^K,\{0,1\}^N)$ with a positive 
stationary distribution $\pi_{\epsilon}$. 
\end{lemma}
\begin{proof}
See Lemma $1$ in \cite{7746656}.
\end{proof}
The stationary distribution of the Markov chain $X_{\epsilon}(t)$ characterizes
 the long term average payoffs of the nodes.
We seek to characterize the stationary distribution of the Markov chain 
$X_{\epsilon}(t)$ for small $\epsilon > 0$.
The following definition helps identify states (the action sequences 
and average payoffs) that 
occur a significant fraction of time, especially, for small $\epsilon$.
 \begin{defn}
 \label{defn stochastically stable states}
 {\cite{RePEc:ecm:emetrp:v:61:y:1993:i:1:p:57-84} Stochastically stable states:}\\
 A state $x \in \Omega$ of a perturbed Markov chain $X_{\epsilon}(t)$ is said to be stochastically stable, if $\lim_{\epsilon \to 0}\pi_{\epsilon}(x)>0$.
\end{defn}

The following theorem characterizes the stochastically stable states of the Markov chain
$X_{\epsilon}(t)$.
\begin{theorem}
\label{thm optimality of algorithm 1}
Under Assumption \ref{assumption Interdependence}, the stochastically stable states of the 
Markov chain $X_{\epsilon}(t)$ are the set of states
that optimize the following formulation:
\begin{align}
\left.\begin{aligned}
\label{formulation2}
\max \;  & \sum_{i =1}^N   U_ i(\bar{r}_i) \\
\text{where,\ \ } \bar{r}_i &= \sum_{a \in \mathcal{A}} p(a) f_i(a) \\
\text{s.t.\ \ \ }  a &=\left(a_1,a_2...a_N\right) \in   \mathcal{A} \\
  p(a) &\in \left\{0, \frac{1}{K}, \frac{2}{K},\cdots, 1\right\}, \sum_{a} p(a) \leq 1
\end{aligned}
\right\}
\end{align}
\begin{proof}(Theorem~\ref{thm optimality of algorithm 1})
See Appendix \ref{Proof of Optimality of alg1}
\end{proof}
\end{theorem}

\begin{remark}
\par A key assumption for G-NUM to work is interdependence. We 
study a completely uncoupled setup where the only feedback to a node on the action 
profile is its payoff. Interdependence ensures that changes in actions by any 
node(s) can be perceived by other nodes in the network as a change in payoff. G-NUM 
exploits this feature where, a discontent node (a node that perceived a change in the 
action profile via a payoff change) changes its action, causing discontent to other 
users in the network. The importance of this assumption is discussed in 
detail in \cite{marden2014achieving}. 
\par In Theorem~\ref{thm optimality of algorithm 1}, we prove that GNUM 
optimizes the formulation in (\ref{utility maximization formulation}) where,
$p(a)$ is restricted to the set $\left\{ 0, \frac{1}{K}, \cdots, \frac{K}{K} \right\}$. For 
large $K$ and bounded utility functions, we note that the performance of our 
proposed algorithm would be approximately optimal (even for small enough 
$\epsilon > 0$).
\end{remark}

\begin{remark}
\label{remark MCMC and G-NUM}
A well known algorithm for combinatorial optimization is a sampling based 
algorithm called simulated annealing introduced by Kirkpatrick et al in
\cite{Kirkpatrick83optimizationby}. Simulated annealing is 
inspired from statistical physics, where samples from the feasible set are obtained 
from a distribution. The distribution has a parameter $T$, called 
the temperature and as $T\to \infty$ the distribution converges to the optimal set. 
\par Markov Chain Monte Carlo (MCMC) based sampling methods have been used to 
achieve such a distribution in a distributed way \cite{4215753}, \cite{6636117}. 
In MCMC based sampling, a reversible Markov chain is constructed to achieve the
distribution in \cite{Kirkpatrick83optimizationby} in a distributed way. 
For constructing such MCMC algorithms in a distributed way, the network is
required to have a Markov random field structure.
In \cite{marden2014achieving} and in G-NUM, samples are obtained from
 a non-reversible Markov chain $X_{\epsilon}$, whose stationary distribution for 
a fixed $\epsilon$ is difficult to characterize.
 However as $\epsilon \to 0$, the stationary distribution converges to the 
 optimal points as seen in Theorem \ref{thm optimality of algorithm 1}. In this 
 approach, we do not require the network to have a Markov random field structure. 
\end{remark}

\section{Distributed subgradient Algorithm for Concave Utility Maximization}
\label{Section: Alg2}
In Remark \ref{remark MCMC and G-NUM}, we stated an important relation
between G-NUM and MCMC based algorithms. MCMC based utility maximization 
algorithms were presented in 
\cite{DBLP:journals/ton/JiangW10,liu2010towards,5625654,6097082}, 
where the parameters of the MCMC algorithm are adapted to achieve utility 
maximization for concave utilities. Inspired by these works,
we propose a completely uncoupled sub-gradient algorithm for concave utility 
maximization.  
\par Throughout this section, we assume that, 
for all $i$, $U_i(\cdot)$ is increasing and strictly concave with $U^{\prime}_i(0)<V$.
 We present our algorithm below, which we call Concave Network Utility 
 Maximization algorithm (C-NUM). As before, every node $i$ is given an internal satisfaction 
 variable $q_i(t)$ taking values $0$ and $1$. The satisfaction variable 
 $q_i(\cdot)$ serves a similar purpose here, where, $q_i(\cdot)=1$ corresponds to 
 node $i$ being ``content'' with the current action chosen and $q_i(\cdot)=0$ 
 represents ``discontent'' state.  The key difference here as compared to G-NUM 
 is that each node is given a weight  $\lambda_i(t)$ taking values in $[0, 
 \lambda_{\max}]$. We divide time into frames of $T$ slots each. The first 
 frame contains slots $1$ to $T$, second frame contains slots $T+1$ to $2T$ and so on. 
 The weight $\lambda_i(t)$ is updated at the end of every frame and is constant 
 during the frame, i.e. $\lambda_i(t)= \lambda_i(l), \;\; t \in~[(l-1)T+1,lT] $. At the 
 end of every frame, node $i$  calculates its  weight $\lambda_i(l)$.  
  However, a node updates its action $a_i(t)$ and 
 satisfaction variable $q_i(t)$ in every time slot.       
     
\par The update of $a_i(t),q_i(t)$ at time $t$ requires only the knowledge of the 
immediate history, i.e. $a_i(t-1),q_i(t-1)$ and $\lambda_i(t)$. (This in contrast to 
G-NUM, where the update at any time required the knowledge of 
previous $K$ actions).  At the beginning of a slot $t$, nodes choose their action 
depending on the satisfaction variable $q_i(t-1)$. 
If node $i$ was content in time slot $t-1$ i.e. $q_i(t-1)=1$, then it repeats the action it 
chose in the previous slot i.e. $a_i(t-1)$ with a large probability $1-\epsilon^c$. Any 
other action is chosen uniformly at random with a probability $\epsilon^c$. In case  
node $i$ was discontent, i.e. $q_i(t-1)=0$,  it chooses an action uniformly at 
random from $\mathcal{A}_i$. Depending on the action profile $a(t)$ chosen, 
node~$i$ receives a payoff $r_i(t)=f_i(a(t))$. 
Node~$i$ updates its satisfaction variable $q_i(t)$  
based on the payoff it received during slot $t$ and its weight $\lambda_i(t)$. In slot $t
$, node~$i$ remains content with probability $1$, if it were content in the previous 
slot, repeats its previous action and the payoff remains unchanged, i.e. $q_i(t-1)=1, $ 
$a_i(t)=a_i(t-1)$ and $r_i(t)=r_i(t-1)$. Else, a node becomes content with a small probability $
\epsilon^{1-\frac{\lambda_i(t) r_i(t)}{\lambda_{\max}}}$ and remains discontent 
with probability $1-\epsilon^{1-\frac{\lambda_i(t) r_i(t)}{\lambda_{\max}}}$. Let 
$s_i(l)$ be the average payoff received by node~$i$ in frame~$l$, i.e. 
\begin{align*}
s_i(l)=\frac{1}{T}\sum_{t=(l-1)T+1}^{lT} r_i(t).
\end{align*}
\par Finally, at the end of frame $l$, $\lambda_i(l+1)$ is updated by a sub gradient algorithm as,
\begin{align}
\bar{r}_i(l) &= \arg\max_{\alpha_i \in [0,1]} U_i(\alpha_i)-\alpha_i \lambda_i(l), \nonumber \\
\label{subgradient update}
\lambda_i(l+1)&= \left[\lambda_i(l)+b(l)\left(\bar{r}_i(l)-s_i(l)\right)\right]^+.
\end{align}  
Here, $\lambda(l+1)$ will serve as the weight for the nodes during frame $l+1$, i.e. from slot $(lT+1)$ to $(l+1)T$
\begin{figure}[!ht]
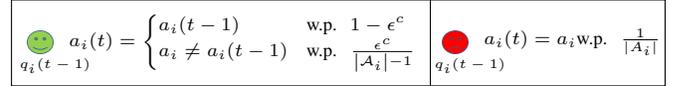
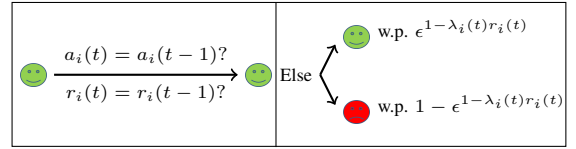

    \centering
    \begin{subfigure}[t]{0.475\textwidth}
   \centering
    \begin{tikzpicture}
\node[inner sep=0](content) at (0,0){\includegraphics[width=.04\textwidth]{content.png}};
 \node[inner sep=0] (C) [below = 0.05cm of content,xshift=0.2cm]{\tiny$q_i(t-1)$};
 \node[inner sep=0](action1)[right = 0.2cm of  content]{{\scriptsize$ a_i(t) =
\begin{cases}
a_i(t-1) & \text{w.p. \ } 1 - \epsilon^c \\
a_i \neq a_i(t-1) & \text{w.p. \ } \frac{ \epsilon^c}{\left|\mathcal{A}_i\right|-1}  
\end{cases}$}};
\node[inner sep=0](discontent) [right = 0.3cm of  action1]{\includegraphics[width=.04\textwidth]{discontent.png}};
 \node[inner sep=0] (D) [below = 0.05cm of discontent,xshift=0.2cm]{\tiny$q_i(t-1)$};
\node[inner sep=0](action2) [right = 0.2cm of  discontent]{{\scriptsize$ a_i(t)=
a_i \text{w.p. \ } \frac{1}{|A_i|}$}};
\draw ($(content.north west)+ (-0.2,0.4)$) rectangle ($(action2.south east)+(0.05,-0.4)$);
\draw ($(action1.north east)+ (0.15,0.18)$) -- ($ (action1.south east)+(0.15,-0.18)$); 
\end{tikzpicture}
\caption{Action Update at time $t$: A content node chooses the action 
that was chosen in the previous slot with a large probability as indicated  
in the left box. 
Whereas, a discontent node chooses an action uniformly 
at random as shown in the right box.}
 \label{Cnum_actionupdate}
    \end{subfigure}

    \begin{subfigure}[t]{0.475\textwidth}
    \centering
    \begin{tikzpicture}
      \node[inner sep=0](content) at (0,0){\includegraphics[width=.04\textwidth]{content.png}};
       \node[inner sep=0](content1) at (3,0){\includegraphics[width=.04\textwidth]{content.png}};
      \draw[->,thick] ($(content.east)+(0.1,0)$) to node[above] {{\scriptsize $a_i(t)=a_i(t-1)? $}} node[below]{ {\scriptsize$r_i(t)=r_i(t-1)?$}} ($(content1.west)+(-0.1,0)$); 
       \node[inner sep=0](else)[right= 0.1 cm of content1] {{\scriptsize Else}};
         \node[inner sep=0](content2) [right = 0.4 cm of else,yshift=0.5 cm]{\includegraphics[width=.04\textwidth]{content.png}};
         \node[inner sep=0](discontent) [right = 0.4 cm of else,yshift=-0.5 cm]{\includegraphics[width=.04\textwidth]{discontent.png}};
         \node[inner sep =0](C)[right=0.1cm of content2,yshift=0.1 cm]{{\scriptsize w.p. $\epsilon^{1-\lambda_i(t)  r_i(t)}$}};
           \node[inner sep =0](D)[right=0.1cm of discontent,yshift=0.1 cm]{{\scriptsize w.p. $1-\epsilon^{1- \lambda_i(t)  r_i(t)}$}};
           \draw[->,thick] ($(else.east)+(0.1,0)$) to ($(content2.west)+(-0.1,-0.1)$); 
           \draw[->,thick] ($(else.east)+(0.1,0)$) to ($(discontent.west)+(-0.1,0.1)$); 
           \draw ($(content.north west)+ (-0.1,0.8)$) rectangle ($(D.south east)+(0,-0.4)$);
           \draw ($(content1.north east)+ (0.05,0.8)$) -- ($ (content1.south east)+(0.05,-0.8)$); 
    \end{tikzpicture}
    \caption{Satisfaction Update at time $t$: A content node becomes content with 
    probability $1$, if it repeats the action that was chosen in the previous slot
     and receives the same payoff that it received in the previous slot, as shown 
     in the left box. In any other case, a node becomes content with 
     probability shown in the right box. Content and discontent states are shown
     as  a green happy and red sad smiley respectively.}
 \label{Cnum_satisfactionupdate}
    \end{subfigure}
    \caption{Cnum update: The Figure \ref{Cnum_actionupdate} shows 
    the update of actions and Figure \ref{Cnum_satisfactionupdate} shows the
    update of satisfaction variable at time $t$.}
    \label{Fig:Cnum}
\end{figure}


\begin{algorithm}
\caption{\textbf{: Concave Network Utility Maximization Algorithm (C-NUM)}}
\label{Alg2}
\begin{algorithmic}
\Initialize{Fix $c > N$, $\epsilon > 0$ and  \\ For all $i$, set $q_i(0) =0, \lambda_i(1)=\lambda_0$.}
\end{algorithmic}

\begin{algorithmic}
\State \textbf{Action update at time $t$:}
\If {($q_i(t-1)=1$)}
\State
$a_i(t) =
\begin{cases}
a_i(t-1) & \text{w.p. \ } 1 - \epsilon^c \\
a_i \in {\mathcal A}_i & \text{w.p. \ } \frac{ \epsilon^c}{\left|\mathcal{A}_i\right|-1} \text{\ if\ } a_i \neq a_i(t-1)
\end{cases}$
\Else
\State $a_{i}(t)= a_i$ w.p. $\frac{ 1}{\left|\mathcal{A}_i \right|}$ where $a_i \in {\mathcal A}_i$
\EndIf
\end{algorithmic}

\begin{algorithmic}
\State \textbf{Update for $q_i(\cdot)$ at time $t$:}
\If { ($q_i(t-1)=1$) and ($a_i(t)=a_i(t-1)$) \\ and $ r_i(t) =  r_i(t-1)$}
\State $q_i(t) = 1 \text{\ w.p. \ } 1$
\Else
\State
$q_i(t) =
\begin{cases}
1 & \text{w.p. \ } \epsilon^{1-\frac{\lambda_i(t) r_i(t)}{\lambda_{\max}}} \\
0 & \text{w.p. \ } 1- \epsilon^{1-\frac{\lambda_i(t) r_i(t)}{\lambda_{\max}}}
\end{cases}$
\EndIf
\end{algorithmic}
\begin{algorithmic}
\State \textbf{Update for $\lambda_i(\cdot)$ at the end of frame $j$:} 
\State $\bar{r}_i(j) = \arg\max_{\alpha_i \in [0,1]} U_i(\alpha_i)-\lambda_i(j)\alpha_i$ 
\State $\lambda_i(j+1)= \left[\lambda_i(j)+b(j)\left(\bar{r}_i(j)-s_i(j)\right)\right]^+$
\end{algorithmic}
\end{algorithm}


\begin{assum}
\label{bounded 1st derivative}
Assume that $U_i^{\prime}(0)<V$ 
\end{assum}
This assumption ensures that the weights $\lambda_i$ are bounded, which is given by the lemma below. A similar assumption is seen in \cite{5625654} (Lemma 19).
\begin{lemma}
\label{lemma bounded dual parameter}
If $\lambda_i(0)< V+1 \; \; \forall i$, then under Assumption \ref{bounded 1st derivative} we have,
 \begin{align*}
\lambda_i(l)\leq V+1 \; \; \forall l
\end{align*}
\begin{proof}
We discuss the proof in Appendix \ref{appendix Proof of lemma bounded dual parameter}.
\end{proof}
\end{lemma}

\subsection{Performance Analysis of C-NUM}
In this subsection, we shall discuss the performance of C-NUM.
We first motivate the basis of C-NUM by formulating the dual problem and showing 
that \eqref{subgradient update} is the approximate sub-gradient update of the 
dual problem. In Lemma \ref{lemma mixing time}, we derive upper bounds for the 
mixing time of the Markov chain induced by C-NUM for a constant $\lambda$. This 
aids in choosing an appropriate frame size. In Theorem \ref{optimality of Alg2}, we 
show that C-NUM is optimal.
\par Consider the optimization problem,
\begin{align}
\begin{aligned}
\label{Formulation alg2}
\max &\sum_i U_i(\bar{r}_i) \\
\text{s.t. } &\bar{r}_i \leq \sum_a p(a) r_i(a) \\ 
&\sum_a p(a) =1, \; p(a)\geq0.
\end{aligned}
\end{align}

The partial Lagrangian of \eqref{Formulation alg2} is given by, 
\begin{align*}
L(\bar{r},p,\lambda)&=\sum_i U_i(\bar{r}_i)  - \sum_i \lambda_i (\bar{r}_i- \sum_a p(a) r_i(a)), \\
&= \sum_i (U_i(\bar{r}_i)- \lambda_i \bar{r}_i) + \sum_a p(a) \sum_i \lambda_i r_i(a).
\end{align*}

The dual of the problem \eqref{Formulation alg2} is 
\begin{align}
\begin{aligned}
\label{dual}
d(\lambda)= \max_{\bar{r},p} L(\bar{r},p,\lambda) \\
\text{s.t. } \sum_a p(a)=1, \;\; p(a)\geq 0
\end{aligned}
\end{align}

The maximization in \eqref{dual} could be split into:
\begin{enumerate}
\item Maximization over $\bar{r}$,
\begin{align} 
\label{flow control problem}
\bar{r} =\max_{\alpha_i\in[0,1]} \sum_i U_i(\alpha_i)- \lambda_i \alpha_i. 
\end{align}

\item Maximization over $p$,
\begin{align}
\begin{aligned}
\label{Max weight problem}
p = \arg \max_{\mu} \sum_a \mu(a) \sum_i \lambda_i r_i(a) \\
\text{s.t. } \sum_a \mu(a)=1, \; \; \mu(a)\geq 0.
\end{aligned}
\end{align}
\end{enumerate}
 Here, \eqref{Max weight problem} is interpreted as the Max-weight 
problem and $\lambda_i$ is interpreted as virtual queue at node $i$.
%
%
By strong duality, the problem in \eqref{Formulation alg2} is equal to the following 
 \begin{align}
 \label{dual min}
 \begin{aligned}
 &\min_{\lambda} d(\lambda) \\
 &\text{s.t. } \lambda_i \geq 0.
 \end{aligned}
 \end{align} 

The sub-gradient of the above dual problem is given by,
\begin{align*}
\sum_a p(a)r_i(a)-\bar{r}_i ,
\end{align*}
where,  $p,\bar{r}$ are the primal optimal solutions.
 
 \begin{remark}The central idea is to use the algorithm in \cite{marden2014achieving}
(which is the same as G-NUM with $K=1$)  to solve the above
 maximization over $p$. Recall that G-NUM induces a Markov chain 
 $X_{\epsilon(t)}$ and Theorem \ref{thm optimality of algorithm 1} 
 characterizes the stationary distribution of $X_{\epsilon(t)}$ as $\epsilon \to 0$. 
However, we need to take care of the following,
\begin{enumerate}
\item the time taken for $X_{\epsilon(t)}$ to converge to its stationary distribution 
\item the effect of using a finite $\epsilon$ in the sub-gradient algorithm
 \end{enumerate}
\end{remark}

  In C-NUM, we use sub-gradient method to solve the dual problem 
 in \eqref{dual min}, where the dual parameters are updated at the 
 end of each frame. For instance, at the end of frame $l$, the approximate 
 subgradient at node $i$ is given by,
\begin{align}
\label{approximate gradient}
s_i(l)-\bar{r}_i(l),
\end{align}
where $s_i(l)$ is the service rate obtained in frame $l$ and $\bar{r}_i(l)$ solves 
\eqref{flow control problem}.
In each frame, where $\lambda$ is kept constant, the algorithm induces an ergodic 
Markov chain, denoted by $X_{\epsilon}$. 
This Markov chain maximizes \eqref{Max weight problem} in the limit 
$\epsilon \to 0$. 
 By using \eqref{approximate gradient} in place of the exact gradient, we need to 
 choose a frame size large enough to ensure that the Markov chain $X_{\epsilon}$ is 
 close to stationarity and $\epsilon$ small enough that \eqref{approximate 
 gradient} is close to the exact gradient.
  
\begin{remark} In \cite{5625654}, Gibbs sampling
is used to solve~\eqref{Max weight problem}, where the Markov random 
field nature of the setup allows a distributed implementation. 
In our setup, usage of such reversible Markov Chain Monte Carlo 
techniques to solve~\eqref{Max weight problem} does not lead to a distributed 
implementation.
  Another key difference here as compared to \cite{5625654}, is the absence of 
variational characterization (See Lemma $21$ in \cite{5625654}). Such 
characterization is possible in \cite{5625654}, due to the structure of 
the stationary distribution of the Gibbs sampling algorithm.
 In C-NUM, we do not have such nice structure in the stationary distribution. In fact,
 we only characterize the stationary distribution as  $\epsilon \to 0$. 
 This leads to an approximate subgradient algorithm even 
 if the Markov chain is assumed to converge instantaneously. 
\end{remark}

In the following discussion, we derive an upper bound for the Mixing time of 
$X_{\epsilon}$ for a fixed $\lambda$. As we shall see, 
this bound provides a trade off between $\epsilon$ and frame size $T$. 
\par Let $\pi_{\epsilon}$ be the stationary distribution of a Markov chain 
$X_{\epsilon}$ and $\pi_t$ be the distribution at time $t$ with $x$ as the initial state, 
i.e. $\pi_0(x)=1$. Define Mixing time of $X$ as,
\begin{align*}
\tau(\zeta)=\min\{t:d_v(\pi_t,\pi_{\epsilon})<\zeta\} \;\; \forall x,
\end{align*}
where,  $d_v(\pi_t,\pi_{\epsilon})=\frac{1}{2} \sum_y |\pi_t(y)-\pi_{\epsilon}(y)|$ is the total variation distance between $\pi_t$ and $\pi_{\epsilon}$.
Now, we have the following lemma on the mixing time of $X_{\epsilon}$ for any $\lambda$.
 \begin{lemma}
\label{lemma mixing time}
\begin{enumerate}
 \item  For C-NUM with any fixed $\lambda$, the mixing time of the Markov chain $X_{\epsilon}(t)$ has the following upper bound,
  \begin{align*}
  \tau(\zeta)<\frac{log(\frac{1}{\zeta})}{\epsilon^{(c+1)N}},
  \end{align*}
\item For G-NUM, the mixing time is upped bounded as, 
\begin{align*}
\tau(\zeta) \leq \left \lceil \frac{\log(\frac{1}{\zeta})}{K \epsilon^{(c+1)NK}} \right \rceil K 
\end{align*}
\end{enumerate}
  \begin{proof}
  See Appendix \ref{Mixing time of Resource allocation Markov chain}
  \end{proof}
\end{lemma}
\begin{remark}
A small $\epsilon$ would provide a better approximation of \eqref{Max weight problem}, however this would lead to a large mixing time and in turn  a large frame size.
\end{remark}

\begin{remark}
Lemma \ref{lemma mixing time} suggests that, for G-NUM, the mixing time grows exponentially 
with $K$. This validates the usage of C-NUM  
against G-NUM for concave utilities, since in G-NUM we increase  
 $K$ for a larger rate region. 
\end{remark}

Now we shall state our main result,
\begin{theorem}
\label{optimality of Alg2}
 Assume $U_i$'s are increasing, strictly concave functions satisfying Assumptions \ref{assumption Interdependence} and \ref{bounded 1st derivative}.
 Then, for fixed frame size $T=\frac{N(V+1)}{\eta\epsilon^{(c+1)N}}$, we have the following, \\
\par $1)$  For step sizes satisfying, 
\begin{align*}
\sum_j b(j)= \infty, \text{ and } \sum_j b^2(j) <\infty.
\end{align*} 
We have,
\begin{align*}
\lim\inf_{L \to \infty} \sum_i U_i(\hat{\bar{r}}_i(L)) \geq  \sum_i U_i(\bar{r}_i^*)  
- \lim\inf_{L \to \infty} \bar{\delta}(L) -\eta,
\end{align*}
and
\begin{align*}
\lim\sup_{t \to \infty} &\sum_i U_i(\bar{r}_i(t)) \geq \sum_i U_i(\bar{r}^*) 
 - \lim\inf_{L\to \infty}\bar{\delta}(L)-\eta,
\end{align*} 
where $\bar{\delta}(L)$ and $\hat{\bar{r}}(L)$ are the cesaro averages given by,  
\begin{align*}
&\bar{b}(L)= \sum_{l=0}^{L-1} b(l), \;\;
\bar{\delta}(L)= \frac{1}{\bar{b}(L)} \sum_{l=0}^{L-1} b(l)\delta(\lambda(l))\\
& \hat{\bar{r}}(L)=\frac{1}{\bar{b}(L)} \sum_{t=0}^{L-1} b(l)\bar{r}(l).
\end{align*}
\par $2)$ For a fixed step size, $b(t)=b\;\; \forall t$, 
 \begin{align*}
 \lim &\inf_{L\to \infty}   \sum_i U_i\left( \frac{1}{L} \sum_{l=0}^{L-1}\bar{r}_i(l)\right) \geq \sum_i U_i(\bar{r}^*) \\
& -\lim\inf_{L\to\infty} \frac{1}{L}\sum_{l=0}^{L-1} \delta(\lambda(l)) - \eta -\frac{Nb}{2},
 \end{align*}
and 
 \begin{align*}
 \lim&\sup_{t\to\infty} \sum_i U_i(\bar{r}_i(t)) \geq \sum_i U_i(\bar{r}^*) \\
 & - \lim\inf_{L \to \infty} \frac{1}{L}\sum_{l=0}^{L-1} \delta(\lambda(l))-\eta-\frac{Nb}{2},
 \end{align*}
where $\bar{r}^*$ is the solution to the following optimization problem,
\begin{align}
\begin{aligned}
\max &\sum_i U_i(\bar{r}_i) \\
\text{s.t. } &\bar{r}_i \leq \sum_a p(a) r_i(a) \\ 
&\sum_a p(a) =1, \; p(a)\geq 0.
\end{aligned}
\end{align}
\begin{proof}
See appendix \ref{proof optimality of alg2}.
\end{proof}
\end{theorem}

\section{Numerical Examples}
\label{section Numerical examples}
\subsection{Illustration}
In this subsection, we illustrate some aspects of C-NUM.
We consider an example with two nodes.  
Each node has two actions to choose from, namely $a_1$ and $a_2$.
The payoff table is shown in Figure \ref{Illustr1}. 
  When nodes choose 
 different actions, the node choosing action $a_2$ gets a higher payoff. However the 
 payoffs are asymmetric, i.e. when node $2$ chooses $a_2$, it gets a 
 payoff of $0.8$ and when node $1$ chooses $a_2$, it gets a payoff of $1$
 (provided the other node chooses $a_1$). In this example, we aim to 
 maximize the sum of log utility, $\sum_i \log(1+\bar{r}_i)$.
\subsubsection{ Effect of frame length~$T$ and number of iterations~$L$}
\par    
\begin{figure}[!ht]
\centering
\captionsetup{justification=centering}
\begin{tabular}{ r|c|c| }
\multicolumn{1}{r}{}
 &  \multicolumn{1}{c}{$a_1$}
 & \multicolumn{1}{c}{$a_2$} \\
\cline{2-3}
$a_1$ & $(0.0001,0.0001)$ & $(0.001,0.8) $ \\
\cline{2-3}
$a_2$ & $(1,0.001)$ & $(0.01,0.01)$ \\
\cline{2-3}
\end{tabular}
\caption{Payoff table ($s_1$ , $s_2$ ) }
\label{Illustr1}
\end{figure}
 In this subsection, we study C-NUM for different frame lengths and  
 iterations of the sub-gradient algorithm.  
  We choose $\epsilon~=~0.01$. We run C-NUM
   for $200$ and $10^6$ frames. 
  Figures~\ref{Plot:Ex1 effect of frame size} and~\ref{Plot:Ex1 effect of frame size 2}
  show the utility of nodes $1$ and $2$ respectively, where C-NUM is run for 
  $200$~frames, with frame lengths $10^4$, $10^6$ and $10^7$ slots. 
   For reference, we plot Max weight and Gibbs sampling based algorithms 
   (requires complete information). 
  In Figure \ref{Plot:Ex1 effect of frame size 3}, we plot the utility 
  of nodes~$1$ and $2$, where C-NUM is run for  $10^6$~frames, with frame lengths of $100$ and $10000$. 
  From the Figures~\ref{Plot:Ex1 effect of frame size}, \ref{Plot:Ex1 effect of 
frame size 2} and \ref{Plot:Ex1 effect of frame size 3}, we observe the following,
\begin{enumerate}
\item For this example, the approximate gradient algorithm converges.
\item For a smaller number of frames $(200)$ 
(See Figure \ref{Plot:Ex1 effect of frame size} and 
\ref{Plot:Ex1 effect of frame size 2}), C-NUM converges  with frame 
lengths of $10^6$ and $10^7$.  
Also the utilities are close to that got by the Gibbs sampling algorithm. 
\item For a larger number of frames, i.e. $(10^6)$, 
even for smaller frame length of $100$ slots, the algorithm converges. 
This is explained by stochastic approximation \cite{borkar2006stochastic}. 
 If the sub-gradient algorithm converges, then convergence is 
 guaranteed irrespective of the frame length. However, the rate of convergence 
 depends on frame length. 
\end{enumerate} 
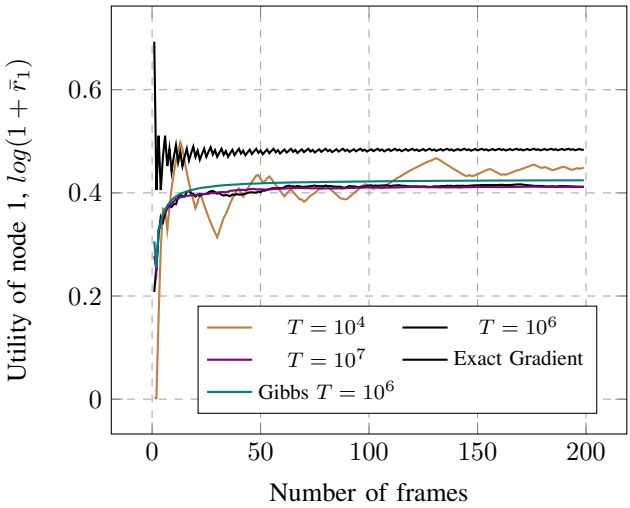
\begin{figure}[ht!]
\begin{center}
\begin{tikzpicture}
\pgfplotstableread{Illustration/example1_e2_001_200_utility.txt} \exa
\pgfplotstableread{Illustration/example1_e4_001_200_utility.txt} \exc
\pgfplotstableread{Illustration/example1_e6_001_200_utility.txt} \exe
\pgfplotstableread{Illustration/example1_e7_001_200_utility.txt} \exf
\pgfplotstableread{Illustration/ex1_exact_grad_200_utility.txt}\exg
\pgfplotstableread{Illustration/gibbs_ex1_e6_001_200_utility.txt}\exh
\begin{axis}[
name=plot1,
grid=major,
grid style={dashed, draw=gray!70},
xlabel={Number of frames},
ylabel={Utility of node 1, $log(1+\bar{r}_1)$},
legend style={at={(0.555,0.3)}, anchor=north,legend columns=2,font=\footnotesize}
]
\addplot +[mark=none,mark size=3,brown,thick]table[x=frames,y=Utility1]\exc;
\addplot +[mark=none,mark size=3,black,thick]table[x=frames,y=Utility1]\exe;
\addplot +[mark=none,mark size=3,violet,thick]table[x=frames,y=Utility1]\exf;
\addplot+[mark=none,mark size=3,thick,black]table[x=frames,y=Utility1]\exg;
\addplot+[mark=none,mark size=3,thick,teal]table[x=frames,y=Utility1]\exh;
\legend{$T=10^4$,$T=10^6$,$T=10^7$, Exact Gradient,Gibbs $T=10^6$ }
\end{axis}
\end{tikzpicture}
\caption{Utility of node $1$ with C-NUM in a simple example with two nodes. 
C-NUM is run for different frame lengths of $100$, $10^4$, $10^6$ and $10^7$. Exact gradient and Gibbs sampling based algorithms are plotted for reference.}
\label{Plot:Ex1 effect of frame size}
\end{center}
\end{figure} 

\begin{figure}[ht!]
\begin{center}
\begin{tikzpicture}
\pgfplotstableread{Illustration/example1_e2_001_200_utility.txt} \exa
\pgfplotstableread{Illustration/example1_e4_001_200_utility.txt} \exc
\pgfplotstableread{Illustration/example1_e6_001_200_utility.txt} \exe
\pgfplotstableread{Illustration/example1_e7_001_200_utility.txt} \exf
\pgfplotstableread{Illustration/ex1_exact_grad_200_utility.txt}\exg
\pgfplotstableread{Illustration/gibbs_ex1_e6_001_200_utility.txt}\exh
\begin{axis}[
name=plot1,
grid=major,
grid style={dashed, draw=gray!70},
xlabel={Number of frames},
ylabel={Utility of node $2$, $log(1+\bar{r}_2)$},
legend style={at={(0.555,0.4)}, anchor=north,legend columns=2,font=\footnotesize}
]
\addplot +[mark=none,mark size=3,brown,thick]table[x=frames,y=Utility2]\exc;
\addplot +[mark=none,mark size=3,black,thick]table[x=frames,y=Utility2]\exe;
\addplot +[mark=none,mark size=3,violet,thick]table[x=frames,y=Utility2]\exf;
\addplot+[mark=none,mark size=3,thick,black]table[x=frames,y=Utility2]\exg;
\addplot+[mark=none,mark size=3,thick,teal]table[x=frames,y=Utility2]\exh;
\legend{$T=10^4$,$T=10^6$,$T=10^7$, Exact Gradient,Gibbs $T=10^6$ }
\end{axis}
\end{tikzpicture}
\caption{Utility of node $2$ with C-NUM, for the simple game with two nodes in Figure 
\ref{Illustr1}. C-NUM is run for different frame lengths of $100$, $10^4$, $10^6$ and 
$10^7$. Exact gradient and Gibbs sampling based algorithms are plotted for 
reference.}
\label{Plot:Ex1 effect of frame size 2}
\end{center}
\end{figure}
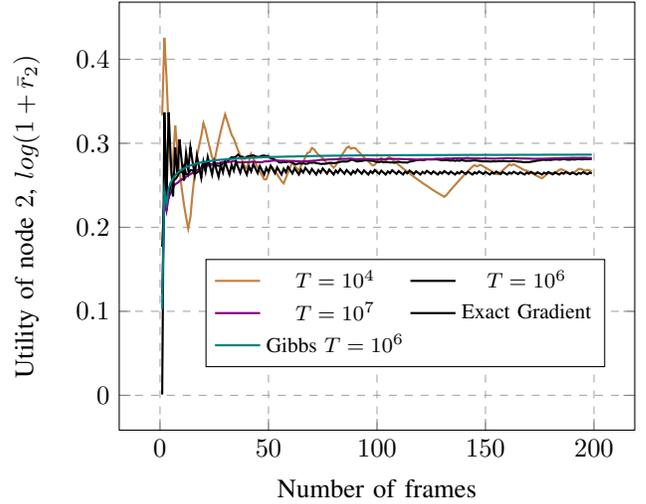 

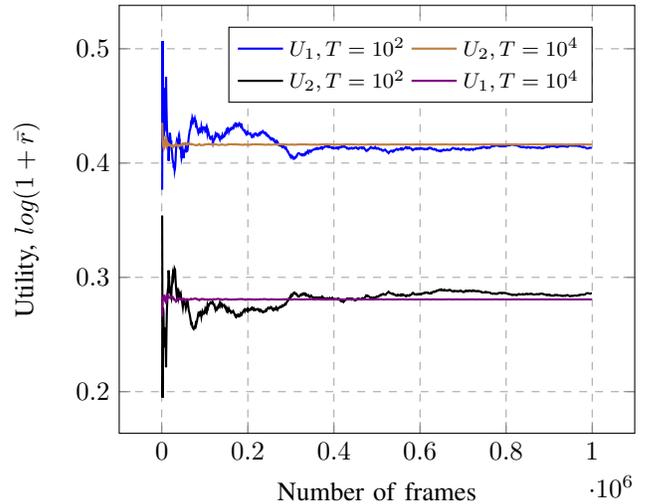
\begin{figure}[ht!]
\begin{center}
\begin{tikzpicture}
\pgfplotstableread{Illustration/example1_e2_001_1000000_utility.txt} \exa
\pgfplotstableread{Illustration/example1_e4_001_1000000_utility.txt} \exc
\pgfplotstableread{Illustration/example1_e6_001_200_utility.txt} \exe
\pgfplotstableread{Illustration/example1_e7_001_200_utility.txt} \exf
\pgfplotstableread{Illustration/ex1_exact_grad_200_utility.txt}\exg
\pgfplotstableread{Illustration/gibbs_ex1_e6_001_200_utility.txt}\exh
\begin{axis}[
name=plot1,
grid=major,
grid style={dashed, draw=gray!70},
xlabel={Number of frames},
ylabel={Utility, $log(1+\bar{r})$},
legend style={at={(0.565,0.95)}, anchor=north,legend columns=2,font=\footnotesize}
]
\addplot+[mark=none,blue,thick] table[x=frames,y=Utility1]\exa;
\addplot +[mark=none,mark size=3,brown,thick]table[x=frames,y=Utility1]\exc;
\addplot +[mark=none,mark size=3,black,thick]table[x=frames,y=Utility2]\exa;
\addplot +[mark=none,mark size=3,violet,thick]table[x=frames,y=Utility2]\exc;
\legend{{$U_1, T=10^2$},{$U_2, T=10^4$},{$U_2, T=10^2$},{$U_1, T=10^4$}}
\end{axis}
\end{tikzpicture}
\caption{Utilities of nodes $1$ and $2$ for the two node game in Figure
\ref{Illustr1}. C-NUM is run for $10^6$ frames, with frame lengths of $100$ and $10000$.}
\label{Plot:Ex1 effect of frame size 3}
\end{center}
\end{figure} 
\subsubsection{ Effect of $\epsilon$}
In this subsection, we study the effect of $\epsilon$ on the performance of 
C-NUM. We run C-NUM for $\epsilon = 0.1,0.01,0.001$ and $0.0001$. 
We fix the  
 frame size as $10^6$~slots and run the algorithm for $200$~frames. 
 In Figure~\ref{Plot:Ex1 effect of epsilon}, we  plot the sum utility for 
 different values of $\epsilon$. 
 We observe that, as $\epsilon $ is decreased from $0.1$ to $0.01$, there is a 
 significant increase in sum utility. However, 
from $\epsilon =0.01$ to $0.001$, there is only a marginal increase in the utility.
 As $\epsilon$ is further reduced to $0.0001$,  the algorithm doesn't 
 converge for the chosen frame size and iteration duration.
 
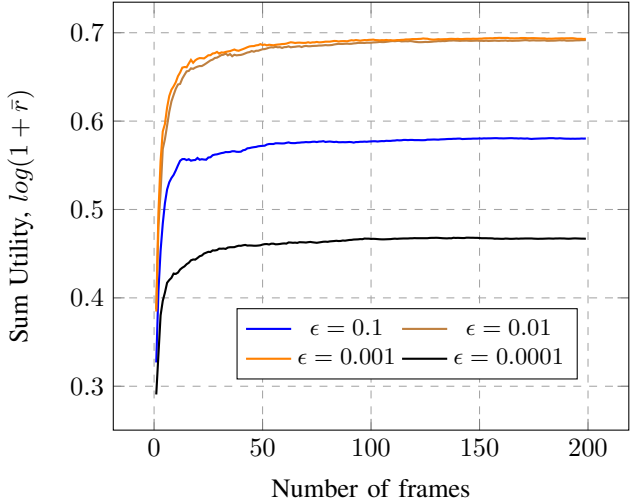
\begin{figure}[ht!]
\begin{center}
\begin{tikzpicture}
\pgfplotstableread{Illustration/example1_e6_1_200_utility.txt} \exa
\pgfplotstableread{Illustration/example1_e6_01_200_utility.txt} \exc
\pgfplotstableread{Illustration/example1_e6_001_200_utility.txt} \exd
\pgfplotstableread{Illustration/example1_e6_0001_200_utility.txt} \exe
\begin{axis}[
name=plot1,
grid=major,
grid style={dashed, draw=gray!70},
xlabel={Number of frames},
ylabel={Sum Utility, $log(1+\bar{r})$},
legend style={at={(0.57,0.285)}, anchor=north,legend columns=2,font=\small}
]
\addplot+[mark=none,blue,thick] table[x=frames,y=SumUtility]\exa;
\addplot +[mark=none,mark size=3,brown,thick]table[x=frames,y=SumUtility]\exc;
\addplot +[mark=none,mark size=3,orange,thick]table[x=frames,y=SumUtility]\exd;
\addplot +[mark=none,mark size=3,black,thick]table[x=frames,y=SumUtility]\exe;
\legend{$\epsilon=0.1$,$\epsilon=0.01$,$\epsilon=0.001$,$\epsilon=0.0001$, $\epsilon=0.00001$ }
\end{axis}
\end{tikzpicture}
\caption{Sum Utility of the average rate with C-NUM for the two node game in \ref{Illustr1}. We run C-NUM for different values of $\epsilon$ varying from $0.1$ to
0.0001}
\label{Plot:Ex1 effect of epsilon}
\end{center}
\end{figure}

\subsection{Example Scenarios} 
In this subsection, we will illustrate G-NUM and C-NUM for two applications 
namely User Association and Channel Selection in WiFi networks.
\subsubsection{User Association}
Here, nodes (players) correspond to the users and 
actions correspond to the set of Access Points (APs). 
The payoff of user $i$ corresponds to the throughput $r_i(t)$.
We consider a fixed IEEE $802.11$ac WiFi network with $2$ Access Points (APs)  and 
$7$ users. The performance for the network
configuration was evaluated using the network simulator ns-3 with the 
following configuration parameters. 
The APs are placed $50$ meters apart from each other. The UEs are 
dropped uniformly around the APs over a square of $50$ meters. 
Each user could associate to either of the two APs. 
We let the APs operate in orthogonal $20$ MHz Channels with a 
maximum achievable throughput of $6.5$ Mbps. 
We consider up-link traffic with saturated queues.
We maximize the sum utility function $\sum_i \log(\delta+\bar{r}_i)$, 
where $\bar{r}_i$ is the average throughput of user $i$. Log utility achieves a 
proportional fair solution \cite{1310314}. Since log utility is unbounded, we 
use $\log(\delta+\bar{r})$ ($\delta>0$).
We plot the normalized sum utility of the users for G-NUM and C-NUM in 
Figure \ref{Plot:User Association} with $\epsilon = 0.2$.
We also plot the Max-Weight based utility maximization algorithm \cite{4455486} 
for reference. 
We observe that the sum-utility of C-NUM and G-NUM is around $0.42$.
The performance is close to the Max-weight based algorithm \cite{4455486} 
which achieves a utility around $0.45$.
\begin{figure}[!ht]
\begin{center}
\begin{tikzpicture}
\pgfplotstableread{geo_utility_recent/User_assoc/0_2Gnum_11.txt} \exg
\pgfplotstableread{geo_utility_recent/User_assoc/01_02Cnum8_3.txt} \exh
\pgfplotstableread{geo_utility_recent/User_assoc/Exactgrad_11.txt}\exi
\begin{axis}[
name=plot1,
xmin=0.0, xmax = 40000000000,
ymin=0.3, ymax=0.475,
grid=major,
grid style={dashed, draw=gray!70},
xlabel={Number of slots},
ylabel={Sum Utility, $\sum_i \log(\delta+\bar{r}_i)$},
legend style={at={(0.65,0.4)}, anchor=north,legend columns=1}
]
\addplot+[mark=none,blue,very thick] table[x=Slots,y=SumUtility]\exg;
\addplot+[mark=none,red,dashed,very thick] table[x=Slots,y=SumUtility]\exh;
\addplot+[mark=none,teal,densely dotted,very thick] table[x=Slots,y=SumUtility]\exi;
\legend{GNUM (K=2),CNUM, Exact Gradient}
\end{axis}
\end{tikzpicture}
\caption{Sum Utility of the nodes for G-NUM and C-NUM for a IEEE 802.11ac WiFi 
network for 7 users and 2 Access points, with $\epsilon=0.2$. The performance of 
Exact gradient algorithm \cite{4455486} is shown for reference.} 

\label{Plot:User Association}
\end{center}
\end{figure}
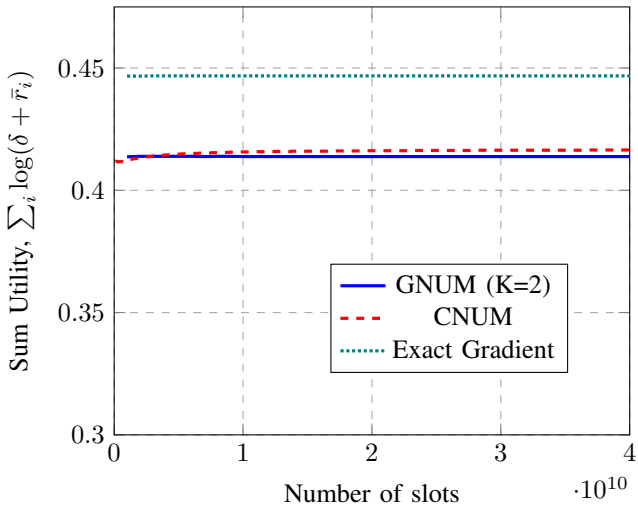 

\subsection{Channel Selection}
In the channel selection problem, nodes (players) 
correspond to ad hoc WiFi links (transmitter-receiver pair) and actions correspond 
to the set of channels that each link could operate. 
The payoff for link $i$ is the throughput $r_i(t)$ it receives.
For the channel selection example, we consider $5$ WiFi 
links (transmitter-receiver pairs) dropped uniformly 
in an area of $100 $ square meters. We assume that each link could operate in 
one of the three $20$ Mhz channels. We fix the network utility as
 $\sum_i \log(\delta+\bar{r}_i)$. We plot the normalized network utility for G-NUM
 and C-NUM in Figure \ref{Plot:Channel Selection}, with $\epsilon = 0.1$. 
 We also plot the Max-weight based algorithm \cite{4455486} for reference. 
 We observe that C-NUM achieves a utility of 0.71 and G-NUM achieves a utility of 
 0.69. The performance is very close to the Max-Weight based algorithm, which 
achieves a utility of 0.73. In comparison to the previous user association example in 
Figure \ref{Plot:User Association}, C-NUM and G-NUM performs better here. This is 
because $\epsilon$ was fixed at $0.2$ in the user association example, 
whereas here it is $0.1$.  

\subsubsection{Key Observations:}  
\begin{enumerate}
 \item [i)] We see that G-NUM and C-NUM perform close to the 
 Max-Weight based algorithm \cite{4455486}. However, there is a small difference 
 in the performance which is due to $\epsilon$. Recall that both G-NUM and C-NUM 
 are optimal as $\epsilon \to 0$. 
\item [ii)] Additionally, we observe that C-NUM outperforms G-NUM for the same 
$\epsilon$. 
To explain this, consider $\Omega^*$, the set of stochastically stable states 
i.e. the set of states having a positive stationary mass as $\epsilon \to 0$. 
We also know that, $\Omega^*$ corresponds to the states having minimum 
stochastic potential (See Definition 3). In G-NUM  $\Omega^*$ is the sequence of 
optimal action sequences and in C-NUM $\Omega^*$ is the actions with maximum 
weight. The stationary mass of $\Omega^*$ is bounded as follows,
\begin{align*}
\pi_{\epsilon}(\Omega^*)&\geq \Theta \left(\frac{|\Omega^*|\epsilon^{\gamma_{\min}}}{|\Omega^*|\epsilon^{\gamma_{\min}}+|\Omega\setminus \Omega^*| \epsilon^{\gamma_2}}\right) \\
\pi_{\epsilon}(\Omega^*)&\leq \Theta\left(\frac{|\Omega^*|\epsilon^{\gamma_{\min}}}{|\Omega^*|\epsilon^{\gamma_{\min}}+|\Omega\setminus \Omega^*| \epsilon^{\gamma_{\max}}}\right)
 \end{align*}
where, $\pi_{\epsilon}$ is the stationary distribution for fixed $\epsilon$; 
$\gamma_{\min}$ and $\gamma_{\max}$ are the minimum and maximum 
stochastic potential respectively, 
$\gamma_2 = \min_{x:\gamma(x)\neq \gamma_{\min}} \gamma(x)$ is the 
second smallest stochastic potential.
In G-NUM, $|\Omega \setminus \Omega^*|$ scales with $K$. This implies, for a fixed 
$\epsilon$, $\pi_{\epsilon}(\Omega^*)$ is closer to $1$
in C-NUM as compared to G-NUM. This explains why C-NUM performs better than 
G-NUM.
\end{enumerate}

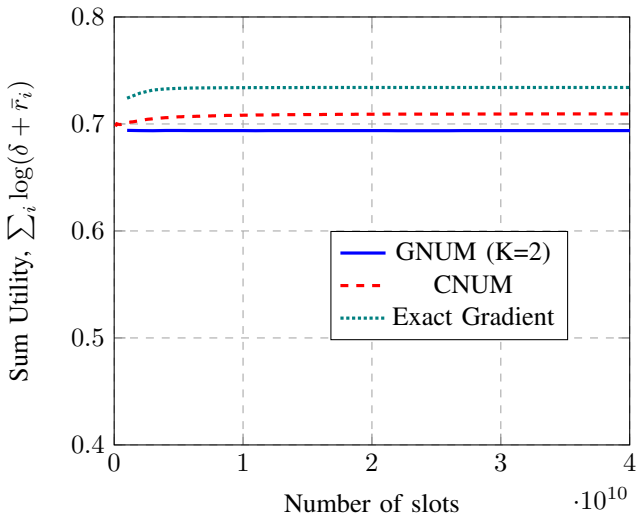
\begin{figure}[!ht]
\begin{center}
\begin{tikzpicture}
\pgfplotstableread{geo_utility_recent/Channel_selec/0_1Gnum_11.txt} \exa
\pgfplotstableread{geo_utility_recent/Channel_selec/1_Cnum_8_3.txt} \exb
\pgfplotstableread{geo_utility_recent/Channel_selec/Exactgrad_11.txt}\exc
\begin{axis}[
name=plot2,
xmin=0.0, xmax = 40000000000,
ymin=0.4, ymax=0.8,
grid=major,
grid style={dashed, draw=gray!70},
xlabel={Number of slots},
ylabel={Sum Utility, $\sum_i \log(\delta+\bar{r}_i)$},
legend style={at={(0.65,0.5)}, anchor=north,legend columns=1}
]
\addplot+[mark=none,blue,very thick] table[x=Slots,y=SumUtility]\exa;
\addplot+[mark=none,red,dashed,very thick] table[x=Slots,y=SumUtility]\exb;
\addplot+[mark=none,teal,densely dotted,very thick] table[x=Slots,y=SumUtility]\exc;
\legend{GNUM (K=2),CNUM, Exact Gradient}
\end{axis}
\end{tikzpicture}
\caption{Sum Utility of the nodes for G-NUM and C-NUM  for a WiFi ad hoc network
for $5$ links and $3$ Channels, with $\epsilon=0.1$. The performance of 
Exact gradient algorithm \cite{4455486} is shown for reference.} 
\label{Plot:Channel Selection}
\end{center}
\end{figure} 

\section{Summary and Comparisons}
\label{section:Summary and Comparisons}
In this section, we will summarize G-NUM and C-NUM.  
G-NUM has the advantage that it maximizes general utilities 
(not necessarily concave). However, we have seen in Lemma~$3$ that the mixing time 
of G-NUM grows exponentially in $K$ (for C-NUM, $K=1$). Also, we have 
seen through simulations that for a fixed $\epsilon$, C-NUM performs better than 
G-NUM.
\par Table~\ref{table_comparison} provides a comparison of G-NUM and C-NUM 
with other distributed algorithms. In Remark~1, 
we stated that the difference between the algorithms 
in \cite{5625654} and our model is the assumption on the network. 
The work in \cite{5625654} assumes a conflict graph model, whereas 
we assume interdependence. 
An important question here is if conflict graph imply interdependence. 
The answer is negative. To see this, we consider a linear network with three links 
where i) links $1$ and $2$ conflict and ii) links $2$
and $3$ conflict. In this case, when link $1$ is transmitting and link $2$ is idle, 
no action change
by link $3$ could change the service rates of links $1$ and $2$. 
This implies that interdependence is not satisfied here. 
 However, this is not a drawback, given that significant change in the 
interference can be used as a signal to indicate that some link has changed its action.  

\begin{table*}[t]
\renewcommand{\arraystretch}{2.5}
\centering
\begin{tabular}{|>{\centering\arraybackslash}p{3.5cm}|>{\centering\arraybackslash}p{2.5 cm}|>{\centering\arraybackslash}p{2.5 cm}|c|  }
 \hline
 Algorithm & Network Model & Utility & Mixing time \\
 \hline
 G-NUM   & Interdependence & General  & $ \left \lceil \frac{\log(\frac{1}{\zeta})}{K \epsilon^{(c+1)NK}} \right \rceil K$\\
 C-NUM & Interdependence    & Concave   &$\frac{log(\frac{1}{\zeta})}{\epsilon^{(c+1)N}}$\\
 Optimal CSMA \cite{5625654}  & Conflict Graph &Concave &$\log(\frac{1}{\zeta})\exp(\Theta(\beta N V))$ \\
 Parallel Glauber Dynamics \cite{6214604} &Conflict Graph with $\max$ degree $\Delta$ & Stabilizes any arrival with rate  $< \frac{1}{\Delta-1}$ & $O(\log N)$ \\
  Pareto Optimality Through Distributed Learning \cite{marden2014achieving} & Interdependence  & Sum-Rate Maximization & $\frac{log(\frac{1}{\zeta})}{\epsilon^{(c+1)N}}$ \\
 \hline
\end{tabular}
\caption{Comparison of G-NUM, C-NUM with other distributed algorithms}
\label{table_comparison}
\end{table*}

\section{Conclusions}
In this work, we have presented two completely uncoupled algorithms for 
utility maximization. This allows a fair allocation of resources, which prior works 
\cite{pradelski2012learning}, \cite{marden2014achieving} in this setup have
ignored. 
\par In the first algorithm, namely G-NUM, we discretize the rate
region, thereby allowing it to be applied for general (possibly non-concave) utilities. 
We show that the set of achievable points in the rate region could be increased by 
increasing the parameter $K$. However, the memory at each node increases as $K$ is 
increased. The state space of the Markov chain increases with $K$; and in Lemma 
\ref{lemma mixing time}, we show that the Mixing time upper bound grows 
exponentially with $K$.
 \par In the second part, we present another algorithm, C-NUM, which is a
 sub-gradient algorithm for concave utilities. In comparison to G-NUM, C-NUM only
  keeps track of the immediate history and the time average service rate. We show 
  convergence in Cesaro averages for decreasing step sizes and time average
  convergence for fixed step size. Through C-NUM, we show an interesting relationship
   between completely uncoupled algorithms and Gibbs sampling based utility 
   maximization algorithms. In future, we would like to compare these algorithms 
   for specific network models of interest.   
\section{Appendix}
\label{section Appendix}
%
\subsection{Proof of Theorem \ref{thm optimality of algorithm 1}:}
\label{Proof of Optimality of alg1}
\par To prove theorem \ref{thm optimality of algorithm 1}, we need to characterize the stationary distribution of $X_{\epsilon}(t)$ for small $\epsilon$. For such a characterization, we shall use the results from \cite{freidlin2012random},\cite{RePEc:ecm:emetrp:v:61:y:1993:i:1:p:57-84} on perturbed Markov chains. 
Let $P_{\epsilon}(x, y)$ denote the transition probability of the Markov chain $X_{\epsilon}$ from state $x$ to state $y$ . Consider the directed graph $\mathcal{G}$ with the states of the Markov chain as vertices and an edge from state $x$ to state $y$ if  $P_{\epsilon}(x , y)>0$. A spanning tree $T_x$ rooted at a vertex $x$ is called a x-tree, i.e. there exist a path from any vertex to $x$ and, $T_x \subseteq \mathcal{G}$ does not contain any cycles. Let the set of all trees rooted at state $x$ be $\mathcal{T}_x$. 

We need the following additional definitions from the theory of regular perturbed Markov processes
from \cite{RePEc:ecm:emetrp:v:61:y:1993:i:1:p:57-84}.
\begin{enumerate}
\item $ \begin{aligned}[t]
\forall x,y \in \Omega,\ \text{if\ } &P^{\epsilon} ( x,y) > 0 \text{ for some } \epsilon > 0, \text{then,}
\end{aligned}$   
\[ 0< \lim_{\epsilon \to 0} \frac{ P^{\epsilon} (x,y)}{\epsilon^{r(x,y)}} < \infty, \]
 $r(x,y)\geq 0$ is defined as the resistance of the transition $x \to y$ (See Definition \ref{defn Regular perturbed Markov Chain}).
\item Consider a sequence of transitions (or a path) $P=x_1 \to x_2 \to ... \to x_k$. 
The resistance of the path is defined as the  sum of the resistances of the one-step 
transitions in the path,
i.e., $r(x_1, x_2)+...+r(x_{k-1}, x_k)$.

\item The resistance from state $x$ to any other state 
$y$, $\rho(x,y)$ is the minimum resistance over all paths 
from $x$ to $y$.
\item The resistance of the tree rooted at $x$, $\rho(x)$ is the sum of the resistances 
of the edges in the tree.
\end{enumerate}
 
Let $\mathcal{T}(a)$ be the set of all trees rooted at state $x$ with action profile sequence $a=(a_1,\cdots,a_K)$. Let $r_{\min}(a)$ be the minimum resistance of all the trees in $\mathcal{T}(a)$. 
\par A state in the Markov chain $X_{\epsilon}$ is of the form, $x=(a,q)$ where 
$a= (a_1,a_2,\cdots,a_K) \in \mathcal{A}^K$ is the sequence of $K$ action profiles 
and $q \in \{0,1\}^N $ is the satisfaction variable of the nodes. 
For any state $x=(a,q)$, we shall use the following definition from 
\cite{RePEc:ecm:emetrp:v:61:y:1993:i:1:p:57-84}, 
\begin{defn}
\label{label stochastic potential}
 Stochastic potential of a state $x$, denoted by $\gamma(x)$,  is the minimum resistance over all the trees rooted at $x$.
 \end{defn}

\begin{lemma}
\label{lemma: states with r_min}
The tree with minimum resistance in $\mathcal{T}(a)$ is rooted at  state $(a,\vec{1})$, i.e. when all the nodes are content. The minimum resistance $r_{\min}(a)$  is given by $c(N-1)+\sum_i 1-U_i\left(\frac{r_i(a(1))+\cdots+r_i(a(K))}{K}\right)$. Further, for any other state $x=(a^{\prime},q)$, where $q\neq \vec{1}$, $r_{\min}(a)<\rho(x),\;\; \forall a\in\mathcal{A}$.
\begin{proof}
A tree $T\in \mathcal{T}(a)$ is rooted at a state $x=(a,q)$, where $q$ could take values in $\{0,1\}^N$. Define the following
\begin{enumerate}
\item[a)] $a^1 = (a,\vec{1})$, where all the nodes are content 
\item[b)] $a^0 = (a,\vec{0})$, where all the nodes are discontent 
\item[c)] $a^q = (a,q)$, where $q$ is a vector with some zeros and some ones (where some nodes are content and some are discontent).
\end{enumerate}

We have the following results:
\begin{enumerate}
\item $\rho(a^1,x) \geq c, \;\; \forall x\neq a^1$. 
For a transition from $a^1$ to a different state to take place, at least one node should change its action. A content node changes its action with probability $\frac{\epsilon^c}{|A_i|-1}$ which has resistance $c$.
\item $\rho(a^1,b^0) = c$. Once a node becomes discontent (which happens with resistance $c$), the other nodes become discontent with zero resistance due to interdependence. 
\item ${\scriptsize \rho(b^0,a^1) = N - \sum_i U_i\left(\frac{f_i(a(1)+\cdots+f_i(a(K))}{K}\right)}; a,b~\in~\mathcal{A} $. 
The resistance for node $i$ to become content is ${\scriptsize1-U_i\left(\frac{f_i(a(1))+...+f_N(a(K))}{K} \right)}$. Since every user must become content from $(b,\vec{0})$, we have the result.
\item $c\leq\rho(a^1,b^1)<2c$. Here at least one node should change its action which happens with resistance $c$. The upper bound follows from the following, $\begin{aligned}[t]\rho(a^1,b^1) \leq \rho(a^1,c^0)+\rho(c^0,b^1)  < 2c\end{aligned}$
\item $\rho(a^q,a^0)=0$. In state $a^q$, some nodes are discontent and due to interdependence, with zero resistance all the nodes become discontent. 
\end{enumerate}
Now, from Lemma $4.3$ from \cite{marden2014achieving}, we have, 
{\small\begin{align}
  \label{potential of content state}
&\gamma(a,\vec{1}) =c(|\mathcal{A}|^K-1)+ N-\sum_iU_i\left(\frac{ r_i(a(1))+\cdots+r_i(a(K))}{K}\right)\\
\label{potential of discontent state}
&\text{Also we have, }\gamma(a,\vec{0})=|\mathcal{A}|^Kc, \; \; \forall a \in \mathcal{A}^K.
\end{align}}
\eqref{potential of discontent state} follows from $5)$ and since every outgoing edge from $a^1$ has resistance $c$ (there are $|\mathcal{A}|^K$ of them). \\     
The stochastic potential of a state $a^q$ is greater than or equal to the stochastic potential of $a^0$, i.e. 
\begin{align}
\label{potential of q state}
\gamma(a^q)\geq|\mathcal{A}|^Kc, \; \; \forall a \in \mathcal{A}^K.
\end{align}
To see \eqref{potential of q state}, let $T_{a^q}$ denote the tree rooted at state $x=(a,q)$ with resistance $\gamma(a^q)$. Due to interdependence, we know that there exists a zero resistance path from $a^q$ to $a^0$. Add the zero resistance path to $T_{a^q}$ and remove all the outgoing edges from $a^0$. This gives us a tree rooted at $a^0$ with no additional resistance. This implies $\gamma(a^q)\geq\gamma(a^0)$. Lemma \ref{lemma: states with r_min} follows from \eqref{potential of content state}, \eqref{potential of discontent state}, \eqref{potential of q state} and noting that $c>N$.
\end{proof}
\end{lemma}

The following theorem from \cite{RePEc:ecm:emetrp:v:61:y:1993:i:1:p:57-84} completes the proof.
\begin{theorem}\cite{RePEc:ecm:emetrp:v:61:y:1993:i:1:p:57-84}. The stochastically stable states of a regular perturbed Markov chain $X_{\epsilon}(t)$ are the states having minimum stochastic potential.
\end{theorem}

The above theorem insists that the stochastically stable states of the Markov chain $X_{\epsilon}(t)$ are the states where all nodes are content and that minimizes $\gamma(x)$, i.e.,
{\scriptsize\[ \sum_i 1-U_i\left(\frac{r_i(a(1))+\cdots+r_i(a(K))}{K}\right). \]}
This implies that the stochastically stable states are those that maximize
{\scriptsize\[ \sum_i U_i\left(\frac{r_i(a(1))+\cdots+r_i(a(K))}{K}\right). \]}
This completes the proof of Theorem~\ref{thm optimality of algorithm 1}.
\qed

\subsection{Proof of Lemma \ref{lemma bounded dual parameter}}
\label{appendix Proof of lemma bounded dual parameter}
Proof is by induction over $l$. The statement is true for $l=0$ (by assumption). Now, we assume that $\lambda_i(l)\in[0, V+1]$. Consider the two cases:
\begin{enumerate}
\item $\lambda_i(l)\leq V. $
In this case, from the update rule, we have
\begin{align*}
\lambda_i(l+1)&=\left[\lambda_i(l)+b(l)\left(\bar{r}_i(l) -r_i(l)\right)\right]^+ \\
&\leq \lambda_i(l) + \bar{r}_i(l) \leq V+1 
\end{align*}
 \item $\lambda_i(l) \in (V \;V+1]$. In this case, 
\begin{align*}
\frac{d}{dy}(U_i(y)-\lambda_i y) 
 \leq U^{\prime}_i(0) - \lambda_i  < 0 \;\; ,\forall  y
\end{align*} 
 The steps follow since $U_i$ is concave and $\lambda_i<V+1$. This implies $U_i(y)-\lambda_i y$ is a decreasing function in $[0,1]$. i.e $\bar{r}_i=0$. $\implies \lambda_i(l+1)\leq V+1$.
\end{enumerate}

\subsection{Mixing time bounds using Dobrushin's inequality}
The Markov chain induced by the algorithm, $X_{\epsilon}(t)$ is a non reversible 
ergodic Markov chain. To analyze the performance of C-NUM, we study the 
mixing time of the Markov chain $X_{\epsilon}(t)$ for a fixed $\lambda$. 
We will use Dobrushin's inequality \cite{opac-b1094914} to derive an upper bound 
on the mixing time. In this section, we will discuss Dobrushin's inequality 
and mixing time based on it. We define 
ergodic coefficient \cite{opac-b1094914} as,
\begin{defn}{Ergodic Coefficient}
The ergodic coefficient of a transition probability matrix $P$, $\delta(P)$ is defined as,
\begin{align*}
\delta(P)=\frac{1}{2} \sup_{i,j} \sum_k |p_{ik}-p_{jk}|
= \sup_{i,j} d_V(p_i,p_j)
\end{align*} 
where, $d_V(\cdot,\cdot)$ is the total variational distance.
\end{defn}
We shall now state the Dobrushin inequality ,
\begin{theorem}{Dobrushin's Inequality:}
Let $P_1$ and $P_2$ be stochastic matrices. Then,
\begin{align*}
\delta(P_1P_2) \leq \delta(P_1) \delta(P_2)
\end{align*}
\begin{proof}
See Theorem $7.1$, Chapter $6$ in \cite{opac-b1094914}
\end{proof}
\end{theorem}
We will now use the above inequality to obtain Mixing time bounds.  
\begin{theorem}
\begin{align*}
d_V(\mu^T P^n,\nu^TP^n) 
 \leq d_V(\mu,\nu) \left(\delta(P)\right)^n
\end{align*}
\begin{proof}
See Theorem $7.2$, Chapter $6$ in \cite{opac-b1094914}
\end{proof}
\end{theorem}
As a corollary, we have,
\begin{align}
\label{bound on tvd}
d_V(\pi_t,\pi) = d_V(\pi_0^TP^t,\pi^TP^t) 
 \leq d_V(\pi_0,\pi) \left(\delta(P)\right)^t
\end{align}
where, $\pi_0$ and $\pi_t$ are the distribution of the Markov chain at times $0$ and
$t$ respectively. The above result indicates that characterizing $\delta(P)$ shall 
provide bounds on the convergence.

\subsection{Proof of Lemma \ref{lemma mixing time}:}
\label{Mixing time of Resource allocation Markov chain}
In this subsection, we shall obtain bounds on the ergodic coefficient $\delta(P)$ and
hence mixing time of the Markov chain $X_{\epsilon}$ for a fixed $\lambda$.
The total variation distance is given by,
\begin{align*}
d_V(p_i,p_j)
= 1-\sum_k p_{ik}\wedge p_{jk} 
\end{align*}
Using the above in the definition of ergodic  coefficient,
\begin{align}
\label{bound on ergodic coeff}
\delta(P) &= 1-\inf_{i,j} \sum_k p_{ik} \wedge p_{jk}  
\leq 1-\sum_k p_{\min,k}
\end{align}
where, $p_{\min,k}=\min_{i} p_{ik}$ is the minimum transition probability to state $k$. 
\par Consider the Markov chain $X_{\epsilon}$ in C-NUM with fixed $\lambda$.
Let $k$ be the a state where all the nodes are content.  
The minimum transition probability to $k$ would correspond to all the $N$ nodes 
becoming discontent with probability $\frac{\epsilon^{cN}}{|\mathcal{A}|}$
and becoming content with probability $\epsilon^{(N-\sum_i \frac{\lambda_i r_i}{\lambda_{\max}})}> \epsilon^N$ (assuming $\epsilon<0.5$).
Thus with a minimum transition probability of 
$\frac{\epsilon^{cN+N}}{|\mathcal{A}|}. $
 Also, note that there are $|\mathcal{A}|$ such transitions. 
Applying the above in \eqref{bound on ergodic coeff}, we have,
\begin{align*}
\delta(P_{\epsilon}) \leq 1- \epsilon^{cN+N} 
\end{align*}
Using \eqref{bound on tvd} from the previous subsection, we have,
\begin{align}
\label{bound on tv distance}
d_V(\pi_t,\pi_{\epsilon}) \leq \left(1- \epsilon^{(c+1)N}\right)^t
\end{align}
From the above, we have,
\begin{align*}
\tau(\zeta) \leq \frac{log(1/\zeta)}{\epsilon^{(c+1)N}}
\end{align*}
\subsubsection{Mixing time bound for G-NUM}
In G-NUM, a state contains $K$ action profiles and only one action profile 
could possibly change in a transition. So, we bound $\delta(P^K)$ instead of 
$\delta(P)$. From the discussion above, we know that,
\begin{align*}
\delta(P^K) = 1- \inf_{i,j} \sum_k p^K_{ik} \wedge p^K_{jk}  
 \leq 1-\sum_k p^K_{\min,k}
\end{align*}
where, $p^K_{ik}$ is the $K$ step transition probability from $i$ to $k$ and 
$p^K_{\min,k}=\min_{i} p^K_{ik}$.
\begin{align*}
\text{Also, }d_V(\pi_t,\pi_{\epsilon}) \leq d_V(\pi_0,\pi_{\epsilon}) \delta(P^K)^{\lfloor \frac{t}{K}\rfloor}
\end{align*}
From any state, the minimum $K$ step transition probability to a state with all the nodes content
is $\frac{\epsilon^{cNK+NK}}{|\mathcal{A}|^K}$. This corresponds to a transition where, in each
step all the nodes becomes discontent with probability $\epsilon^c/|\mathcal{A}|$ and 
becomes content with probability $\epsilon^N$ and there are $K$ such transitions.
Thus for G-NUM, we have,
\begin{align*}
d_V(\pi_t,\pi_{\epsilon}) \leq (1- \epsilon^{(c+1)NK})^{\lfloor \frac{t}{K} \rfloor K} 
\end{align*}
and the mixing time is bounded by,
{\small \begin{align*}
\tau(\zeta) \leq \left\lceil \frac{\log(1/\zeta)}{K \epsilon^{(c+1)NK}} \right\rceil K 
\end{align*}
}
\qed

\subsection{Proof of theorem \ref{optimality of Alg2}:}
\label{proof optimality of alg2}
The proof follows the standard approximate subgradient algorithm in
\cite{doi:10.1137/S1052623400376366}, if we assume that the Markov chain
$X_{\epsilon}$ converges to its stationary distribution while updating the weights in
\eqref{subgradient update}. We follow the analysis in \cite{5625654} except that the update in \eqref{subgradient update} with $s_i(t)$ replaced by the 
payoff averaged over the stationary distribution of $X_{\epsilon}$ is an approximate
subgradient.
\par Let $\delta(\lambda)$ denote the error in the subgradient, i.e.
\begin{align}
\delta(\lambda) =  \max_a \sum_i r_i(a) \lambda_i - \sum_a p(a,\lambda) \sum_i r_i(a) \lambda_i \nonumber\\
\label{error due to Marden young}
\implies \sum_i s_i(\lambda) \lambda_i = \max_a \sum_i r_i(a) \lambda_i -\delta(\lambda),
\end{align}
 where, $p(a,\lambda)$ is the stationary distribution of the Markov chain 
($X_{\epsilon}$) and $s_i(\lambda) = \sum_a p(a,\lambda) r_i(a)$ is the service rate 
obtained with fixed $\lambda$. 
Since,
\begin{align}
&\bar{r}_i(l) = arg\max_{\alpha\in [0,1]} U_i(\alpha)-\alpha\lambda_i(l), \nonumber\\
\label{utility vs optimal}
&\text{we have,  }U_i(\bar{r}_i(l)) - \bar{r}_i(l) \lambda_i(l) \geq U_i(\bar{r}_i^*) - \bar{r}_i^* \lambda_i(l) \\
&\text{where, $r^*$ is the optimal solution of \eqref{Formulation alg2}. } \nonumber \\
&\text{Also, } \sum_i \bar{r}_i^* \lambda_i(l) \leq \max_a \sum_i r_i(a) \lambda_i(l). \nonumber
\end{align}
Substituting the above in \eqref{utility vs optimal} and summing over $i$, we get,
\begin{align}
\label{utility bounds}
\sum_i \bar{r}_i(l) \lambda_i(l) \leq \sum_i U_i(\bar{r}_i(l)) - U_i(\bar{r}_i^*) + \max_a \sum_i r_i(a) \lambda_i(l).
\end{align}
From \eqref{error due to Marden young} and \eqref{utility bounds}, we have, 
\begin{align}
\label{error due to Marden young and utility bounds}
\begin{aligned}
2 b(l)\sum_i \lambda_i(l)  \left(\bar{r}_i-s_i(\lambda(l))\right) \leq  2b(l)\sum_i U_i(\bar{r}_i(l))& \\
 - 2b(l) U_i(\bar{r}_i^*) +2b(l) \delta(\lambda)&.
\end{aligned}
\end{align}
 Then, for node $i$, we have, 
\begin{align}
\lambda_i^2(l+1) = &\left(\left[\lambda_i(l)+ b(l)\left(\bar{r}_i(l)-  s_i(l)\right)\right]^+\right)^2 \nonumber \\
\lambda_i^2(l+1) \leq& \lambda^2_i(l) + 2b(l) \lambda_i(l)\left( \bar{r}_i(l)-  s_i(l)) \right) \nonumber\\
&+b^2(l)(\bar{r}_i(l)-s_i(l))^2 \nonumber\\
\label{drift bound}
\leq & \lambda_i^2(l) + 2 b(l)\lambda_i(l) \left(\bar{r}_i(l)-s_i(l)\right)+ b^2(l).
\end{align}
Summing $\eqref{drift bound}$ over all the nodes, we get,
\begin{align}
\label{eqn sumdrift}
\begin{aligned}
\sum_i \lambda_i^2(l+1)) \leq& \sum_i \lambda_i^2(l) +N b^2(l) + 2b(l) e(l)   \\
&+2b(l) \sum_i \lambda_i(l) \left(\bar{r}_i(l)-s_i(\lambda(l))\right),
\end{aligned}
\end{align}
where, $e(l)=  \sum_i \lambda_i(l)\left(s_i(\lambda(l))-s_i(l)\right) $ is the error 
due to the fact that the Markov chain $X_{\epsilon}$ has not converged to 
its stationary distribution.
Substituting \eqref{error due to Marden young and utility bounds} in \eqref{eqn sumdrift}, we get,
\begin{align}
\sum_i \lambda_i^2(l+1)) \leq & \sum_i \lambda_i^2(l) +N b^2(l) +2b(l) e(l) +2b(l) \delta(\lambda)   \nonumber \\
\label{drift in terms of utility}
 &  + 2b(l)\sum_i \left(U_i(\bar{r}_i(l)) -  U_i(\bar{r}_i^*)\right). 
\end{align}

Next we shall consider two cases step sizes of $b(l)$. 

\subsubsection{Decreasing step size}
Choose $b(l)$ such that,
\begin{align*}
&\sum_t b(l) = \infty, \; \text{ and } \sum_t b^2(l)< \infty \\
&\text{We define the following Cesaro averages, } \\
&\bar{b}(L)= \sum_{l=0}^{L-1} b(l), \; \bar{U}(L)= \frac{1}{\bar{b}(L)} \sum_{l=0}^{L-1} b(l) \sum_i U_i(\bar{r}_i(l)),\\
&\bar{\delta}(L)= \frac{1}{\bar{b}(L)} \sum_{l=0}^{L-1} b(l)\delta(\lambda(l)), \; \hat{\bar{r}}(L)=\frac{1}{\bar{b}(L)} \sum_{l=0}^{L-1} b(l)\bar{r}(l).
\end{align*}
Summing \eqref{drift in terms of utility} from $l=0$ to $l=L-1$ and normalizing, 
\begin{align}
&\frac{1}{\bar{b}(L)}\sum_{l=0}^{L-1} \sum_i \lambda_i^2(l+1) \leq  \frac{1}{\bar{b}(L)} \sum_{l=0}^{L-1}  \sum_i \lambda_i^2(l) \nonumber\\ 
&+ \frac{2}{\bar{b}(L)}\sum_{l=0}^{L-1}\sum_i \left( b(l) U_i(\bar{r}_i(l)) - b(l) U_i(\bar{r}_i^*)\right) \nonumber \\
  &+ \frac{2}{\bar{b}(L)}\sum_{l=0}^L \left( b(l)\delta(\lambda(l)) +N b^2(l) +b(l) e(l)\right)\nonumber \\
\implies &\frac{1}{2\bar{b}(L)} \sum_i \lambda_i^2(L)- \lambda_i^2(0)  \leq  \bar{U}(L) -  \sum_i U_i(\bar{r}_i^*) \nonumber \\
\label{cesaro average drift bound}
 &  +\bar{\delta}(\lambda(l)) +\frac{N}{\bar{b}(L)} \sum_{l=0}^{L-1} b^2(l)  +  \frac{1}{\bar{b}(L)} \sum_{l=0}^L b(l) e(l) \\
&\text{Consider, }
\frac{1}{\bar{b}(L)} \sum_{l=0}^L b(l) e(l) = \frac{1}{\bar{b}(L)} S(L) + E(e(l)|\mathcal{F}_l), \nonumber
\end{align}
where, $S(L)=\sum_{l=0}^{L-1} b(l) (e(l) -E(e(l)|\mathcal{F}_l))$ is an $\mathcal{F}_L$ martingale. Here $\mathcal{F}_l$ denotes the filtration until frame $l$.
 \begin{align*}
  &E(S(L)-S(L-1))^2 =  E(b(L)(e(L) -E(e(L)|\mathcal{F}_L)))^2 \\
 &\implies \sum_L E(S(L)-S(L-1))^2 \leq N(V+1)\sum_L b^2(L) 
 < \infty
 \end{align*}
 
 By martingale convergence theorem \cite{williams1991probability}, $\lim_{L\to \infty}S(L)$ converges $a.s$. 
 This implies, 
 \begin{align}
 \label{martingale convergence thm}
\lim_{L\to \infty} \frac{1}{\bar{b}(L)} \sum_{l=0}^L b(l)(e(l) -E(e(l)|\mathcal{F}_l)) =0
 \end{align}
 \begin{align*}
&E(e(l)|\mathcal{F}_l) \stackrel{(a)}{\leq} (V+1) \sum_i  E(s_i(l)|\mathcal{F}_l) - s_i(\lambda(l)) \\
&= \frac{V+1}{T}  \sum_{t=(l-1)T}^{lT} \sum_{i,a} \left(r_i(a) \pi(a,t) - r_i(a) \pi(a,\lambda(l))\right)  \\
& \stackrel{(b)}{\leq} \frac{N(V+1)}{T} \sum_{t=(l-1)T}^{lT} d_v(\pi(a,t),\pi(\lambda(l))) 
 \stackrel{(c)}{\leq} \frac{1}{T} \frac{N(V+1)}{\epsilon^{cN+N}}, 
 \end{align*}
 where, $(a)$  and $(b)$ follows, since for all $i$, $\lambda_i(t)<V+1$ and 
 $r_i\leq1$ ; $(c)$ follows from the mixing time bound 
 in \eqref{bound on tv distance}. 
\par By our choice of frame size $T=\frac{N(V+1)}{\eta\epsilon^{cN+N}}$, we have,
 \begin{align}
 \label{mixing bound}
 E(e(l)|\mathcal{F}_l)<\eta
 \end{align}
 
Taking limit $L \to \infty$, in \eqref{cesaro average drift bound} and 
using \eqref{martingale convergence thm}, \eqref{mixing bound} we get,
\begin{align*}
\lim\inf_{L \to \infty} \bar{U}(L) \geq& \sum_i U_i(\bar{r}_i^*) 
 - \lim\inf_{L \to \infty} \bar{\delta}(L) - \eta
\end{align*}
Also by Jensen's inequality, we have,
\begin{align*}
\sum_i U_i(\hat{\bar{r}}_i(L)) &\geq \bar{U}(L), \;\; \forall L.  \\
\implies \lim\inf_{L \to \infty}  \sum_i U_i(\hat{\bar{r}}_i(L)) &\geq \sum_i U_i(\bar{r}_i^*) 
- \lim\inf_{L \to \infty}  \bar{\delta}(L) - \eta.
\end{align*}
The above algorithm is an approximate sub-gradient method discussed in \cite{doi:10.1137/S1052623400376366}. By lemma $2.1$ in \cite{doi:10.1137/S1052623400376366}, we have,
\begin{align*} 
\lim\sup_{l \to \infty}  \sum_i U_i(\bar{r}_i(l)) \geq  \sum_i U_i(\bar{r}_i^*) + \lim\inf_{L \to \infty} \bar{\delta}(L) + \eta.
\end{align*} 
\subsubsection{Fixed step size ($b(l)=b)$}
When $b(l)$ is a constant,
\begin{align*}
\sum_i &\lambda_i^2(l+1) \leq \sum_i \lambda_i^2(l) + 2b\sum_i \left(U_i(\bar{r}_i(l)) -  U_i(\bar{r}_i^*)\right)  \\
 &+2b \delta(\lambda(l)) +N b^2  
 +2b e(l) \text{ (Using \eqref{drift in terms of utility}). }
\end{align*} 
Averaging from $l=0$ to $L-1$,  we get,
\begin{align}
& \frac{1}{L}\sum_{l=0}^{L-1}\sum_i \lambda_i^2(l+1) \leq \frac{1}{L}\sum_{l=0}^{L-1}\sum_i \lambda_i^2(l) + \frac{2b}{L} \sum_{l=0}^{L-1} \delta(\lambda(l))+Nb^2 \nonumber\\ 
&+ \frac{2b}{L} \sum_{l=0}^{L-1} e(l)+\frac{2b}{L}\sum_{l=0}^{L-1} \sum_i U_i(\bar{r}_i(l))- 2b\sum_i U_i(\bar{r}^*).  \nonumber \\
&\implies \frac{1}{2bL} \sum_i (\lambda^2_i(L)-\lambda^2_i(0)) \leq  \frac{1}{L}\sum_{l=0}^{L-1} \sum_i U_i(\bar{r}_i(l))\nonumber \\
\label{drift bound fixed b}
& - \sum_i U_i(\bar{r}^*)  
 + \frac{1}{L} \sum_{l=0}^{L-1} \delta(\lambda(l))+\frac{Nb}{2} 
+ \frac{1}{L} \sum_{l=0}^{L-1} e(l). \\
&{\small\text{Consider,}  \frac{1}{L}\sum_{l=0}^{L-1} e(l)  = \frac{1}{L} \sum_{l=0}^{L-1} (e(l) -E(e(l)|\mathcal{F}_t)) +  E(e(l)|\mathcal{F}_t), \nonumber
}
 \end{align}
 
Recall $S(L)=\sum_{l=0}^{L-1} \frac{(e(l)-E(e(l)|\mathcal{F}_t))}{l}$, with $b(l)=\frac{1}{l}$.
We know that $S(L)$ converges a.s. By Kroneker's lemma \cite{williams1991probability}, we have, almost surely,
\begin{align*}
\lim_{L\to\infty} \frac{1}{L} \sum_{l=0}^{L-1} (e(l) -E(e(l)|\mathcal{F}_t)) =0
\end{align*} 
Also by our choice of step size, we have,
$E(e(l)|\mathcal{F}_t))<\eta $
Now, using the above and taking limit $L \to \infty$ in \eqref{drift bound fixed b}, 
\begin{align*}
\lim &\inf_{L\to \infty}   \frac{1}{L} \sum_{l=0}^{L-1}\sum_i U_i(\bar{r}_i(l)) \geq \sum_i U_i(\bar{r}^*) \\
& -\lim\inf_{L\to\infty} \frac{1}{L}\sum_{l=0}^{L-1} \delta(\lambda(l)) - \eta -\frac{Nb}{2}.
 \end{align*}
 By Jensen's inequality, we get,
 \begin{align*}
 \lim &\inf_{L\to \infty}   \sum_i U_i\left( \frac{1}{L} \sum_{l=0}^{L-1}\bar{r}_i(l)\right) \geq \sum_i U_i(\bar{r}^*) \\
& -\lim\inf_{L\to\infty} \frac{1}{L}\sum_{l=0}^{L-1} \delta(\lambda(l)) - \eta -\frac{Nb}{2}.
 \end{align*}
 By lemma $2.1$ in \cite{doi:10.1137/S1052623400376366}, we have,
 \begin{align*}
 \lim&\sup_{l\to\infty} \sum_i U_i(\bar{r}_i(l)) \geq \sum_i U_i(\bar{r}^*) \\
 & - \lim\inf_{L \to \infty} \frac{1}{L}\sum_{l=0}^{L-1} \delta(\lambda(l))-\eta-\frac{Nb}{2}.
 \end{align*}
\bibliography{mybib_journal}
\bibliographystyle{IEEEtran}
\end{document}